%% file: freyd-bicats.tex
\documentclass[submission,copyright,creativecommons]{eptcs}
\providecommand{\event}{ACT 2023} %

\input{lics-preamble}

\usepackage{subcaption}

\usepackage[normalem]{ulem} %
\usepackage{scalerel} %

\newcommand{\V}{\mathcal{V}}
\newcommand{\J}{\mathsf{J}}
\newcommand{\FreydBicat}[1]{\mathbf{FreydBicat}}
\newcommand{\FreydAct}[1]{\mathbf{FreydAct}}
\newcommand{\lefttrans}{{\vartheta}}
\newcommand{\righttrans}{{\chi}}
\newcommand{\centre}{\mathcal{Z}}

\renewcommand{\d}{\mathrm{d}}

\usepackage{extarrows}
\usepackage{placeins}

\newcommand{\actrho}{\widetilde{\rho}}
\newcommand{\ractrho}{\widetilde{\rho^{\ract}}}

\newcommand{\lactlambda}{\toact{\lambda^\lact}}
\newcommand{\lactlambdaprime}{\toact{\lambda^{\!\lactprimesmall}}}

\newcommand{\ractlambda}{\ractrho}

\newcommand{\lactalpha}{\toact{\alpha^\lact}}
\newcommand{\lactalphaprime}{\toact{\alpha^{\lactprimesmall}}}
\newcommand{\ractalpha}{\toact{\alpha^\ract}}

\newcommand{\kappaone}{{\kappa}}
\newcommand{\kappaprime}{{\kappa'}}

\newcommand{\toact}[1]{\widetilde{#1}}
\newcommand{\lactprimesmall}{\scaleobj{.8}{\blacktriangleright}}

\DeclareMathOperator{\lactprime}{\lactprimesmall}

\usepackage{iftex}

\ifpdf
  \usepackage{underscore}         %
  \usepackage[T1]{fontenc}        %
\else
  \usepackage{breakurl}           %
\fi
\usepackage{wrapfig}
\title{Effectful Semantics in 2-Dimensional Categories: \\ Premonoidal and Freyd Bicategories}
\author{Hugo Paquet
\institute{LIPN,  Université Sorbonne Paris Nord \\ 
Paris, France}
\email{paquet@lipn.univ-paris13.fr}
\and
Philip Saville
\institute{University of Oxford\\
Oxford, UK}
\email{philip.saville@cs.ox.ac.uk}
}
\def\titlerunning{Premonoidal and Freyd Bicategories}
\def\authorrunning{H. Paquet, P. Saville}

\begin{document}
\maketitle

\begin{abstract}
  Premonoidal categories and Freyd categories provide an encompassing
  framework for the semantics of call-by-value programming languages.
  Premonoidal categories
  are a weakening of monoidal categories in which the interchange law
  for the tensor 
  product 
  may not hold,
  modelling the fact that
  effectful programs cannot generally be re-ordered. A Freyd category
  is a pair of categories
  with the same objects: 
  a premonoidal category of general programs, and a monoidal category of `effect-free' programs which do admit re-ordering.

Certain recent innovations in semantics, however, have produced  
models which are not categories but bicategories. 
Here we develop the theory to capture such examples by introducing premonoidal and Freyd structure in a bicategorical setting. 
The second dimension introduces new subtleties, so we 
verify our definitions with several examples and a correspondence theorem---between Freyd bicategories and certain actions of monoidal bicategories---which parallels the categorical framework. 
\end{abstract}

\section{Introduction}
\label{sec:introduction}
\input{introduction}

\label{sec:bicat-theory}

\section{Premonoidal bicategories}
\label{sec:premonoidal-bicategories}

\input{premonoidal}

\section{Freyd bicategories}
\label{sec:freyd-bicategories}

\input{sec-freyd-bicategories}

\section{Freyd bicategories and actions}
\label{sec:freyd-bicats-and-actions}

\label{sec:actions}

\input{actions}

\section{Conclusions}
\label{sec:conclusion}
\input{sec-conclusion}

\paragraph*{Acknowledgements.} 
HP was supported by a Royal Society University Research Fellowship and by a Paris Region Fellowship co-funded by the European Union (Marie Sk\strokeL{}odowska-Curie grant agreement 945298).
PS was supported by the 
Air Force Office of Scientific Research under award number {FA9550-21-1-0038}.
Both authors thank D. McDermott, N. Arkor, and the Oxford PL
group for useful discussions.

\bibliography{short-bib}
\bibliographystyle{eptcs}

\appendix

\section{Missing coherence axioms for \texorpdfstring{\Cref{def:freyd-action}}{Definition 22}}
\label{app:missingaxioms}

We complete the list of coherence equations for
\Cref{def:freyd-action}. In the diagrams that follow, we consider $b : B \to B'$
and $c : C \to C'$ in $\B$ and, for $f : A \to A'$ and $g : B \to B'$
in $\V$, we write $\nu$ for the composite $\J f \ract g \XRA{\zeta}
\J(f \tens g) \XRA{\theta} f \lact \J g$.
	In \Cref{res:correspondence-theorem}, $\kappa$ represents the
pseudonaturality of the associator $\alpha$ in its middle argument,
and these axioms enforce the appropriate modification conditions
(\cf~\Cref{fig:premonoidal-bicat-modifications}).

\input{diag-extra-diagrams-for-def-of-Freyd-action}

\section{Proofs for  \texorpdfstring{\Cref{sec:freyd-bicats-and-actions}}{Section 4}}
\input{sec-action-proofs}

\end{document}

%% file: lics-preamble.tex
\usepackage{quiver}
\usepackage{graphicx}
\usepackage{todonotes}

\usepackage{cmll}
\usepackage{amsthm}
\usepackage{amsmath}
\usepackage{enumitem}
\usepackage{microtype} 
\usepackage{cleveref}
\usepackage{mathtools} 
\usepackage{enumitem}
\usepackage{tikz}
\usepackage{mathrsfs}
\usepackage{stmaryrd}
\usepackage{textcomp} %
\usepackage{wrapfig}
\usepackage{caption}
\usepackage{booktabs}
\usepackage{makecell}
\usepackage{framed}
\tikzcdset{scale cd/.style={every label/.append style={scale=#1},
    cells={nodes={scale=#1}}}}
\tikzcdset{scalenodes/.style=
  {cells={nodes={scale=#1}}}}

\let\strokeL\l

\usepackage{thm-restate}

\usepackage{listings}
\definecolor{darkgreen}{RGB}{0,102,0}
\definecolor{lightblue}{RGB}{111, 151, 242}
\definecolor{darkred}{RGB}{178,0,0}
\definecolor{lightgrey}{rgb}{0.5,0.5,0.5}
\definecolor{mymauvnne}{rgb}{0.58,0,0.82}
\definecolor{brightyellow}{RGB}{237, 217, 83}

\definecolor{orange}{RGB}{255,123,0}
\definecolor{lightpink}{RGB}{255, 103, 129}
\definecolor{brightpink}{RGB}{232, 88, 232}
\definecolor{grassgreen}{RGB}{0,154,23}
\definecolor{seablue}{RGB}{52,111,111}
\definecolor{skyblue}{RGB}{135, 206, 235}
\newcommand{\opacity}{0.3}

\newcommand{\alias}[1]{ |[alias=#1]| }

\newcommand{\monoidalcolour}							 {lightpink}
\newcommand{\strongpseudofuncolour}					{grassgreen}

\newcommand{\extcolour}											{skyblue}
\newcommand{\actioncolor}										{darkred}
\newcommand{\naturalitycolour}								{white}

\newcommand{\monoidalpcolour}{\monoidalcolour}
\newcommand{\monoidalmcolour}{\monoidalcolour}
\newcommand{\monoidallcolour}{\monoidalcolour}
\newcommand{\monoidalrcolour}{\monoidalcolour}

\newcommand{\actmcolour}{\actioncolor}
\newcommand{\actlcolour}{\actioncolor}

\newcommand{\actpcolour}{\actioncolor}
\newcommand{\kcolour}{\strongpseudofuncolour}

\newcommand{\Klambdacolour}{\naturalitycolour}
\newcommand{\actlambdacolour}{\naturalitycolour}

\newtheorem{theorem}{Theorem}
\newtheorem{definition}[theorem]{Definition}
\newtheorem{proposition}[theorem]{Proposition}

\newtheorem*{lemma*}{Lemma}
\newtheorem*{proposition*}{Proposition}
\newtheorem*{theorem*}{Theorem}

\input{macros}

\newcommand{\oncell}[1]{  #1 }

\newcommand{\B}{\mathscr{B}}
\newcommand{\C}{\mathscr{C}}



%% file: macros.tex
\newcommand{\vertequals}{\rotatebox{90}{$\,=$}}
\newcommand{\para}[1]{\widetilde{#1}}

\newcommand{\homBicat}[2]{ \mathrm{Hom}({#1}, {#2}) }

\newcommand{\cf}{\emph{c.f.}}

\newcommand{\eg}{\emph{e.g.}}
\newcommand{\ie}{\emph{i.e.}}

\newcommand{\st}{\mid} 

\newcommand{\K}{\mathcal{K}}

\newcommand{\To}{\ensuremath{\Rightarrow}} 
\newcommand{\xra}[1]{\ensuremath{\xrightarrow{#1}}} 
\newcommand{\XRA}[1]{\ensuremath{\xRightarrow{#1}}}
\newcommand{\id}{\ensuremath{\mathrm{id}}} 
\newcommand{\Id}{\mathsf{Id}} 

\renewcommand{\epsilon}{\varepsilon}

\newcommand{\psinv}[1]{{#1}^{{\bullet}}} 
\newcommand{\iso}{\cong}

\newcommand{\Nat}{\mathbb{N}}

\newcommand{\act}{\triangleright}
\newcommand{\lact}{\triangleright}
\newcommand{\ract}{\triangleleft}
\newcommand{\tens}{\otimes}
\newcommand{\tensu}{I}
\newcommand{\cellOf}[1]{\overline{#1}}

\newcommand{\monoidal}[1]{ (#1, \tens, \tensu) }

\newcommand{\pentagonator}{\mathfrak{p}}
\newcommand{\montrianglel}{\mathfrak{l}}
\newcommand{\montriangler}{\mathfrak{r}}
\newcommand{\montrianglem}{\mathfrak{m}}

\newcommand{\actalpha}{{\widetilde\alpha}}
\newcommand{\actlambda}{{\widetilde\lambda}}
\newcommand{\actpentagonator}{{\widetilde\pentagonator}}
\newcommand{\acttrianglel}{{\widetilde\montrianglel}}
\newcommand{\acttrianglem}{{\widetilde\montrianglem}}
\newcommand{\acttriangler}{{\widetilde\montriangler}}

\newcommand{\catC}{\ensuremath{\mathbb{C}}} 
\newcommand{\catD}{\ensuremath{\mathbb{D}}} 
\newcommand{\catE}{\ensuremath{\mathbb{E}}} 

\newcommand{\catV}{\mathbb{V}}

\newcommand{\baseCat}{\B}
\newcommand{\Set}{\mathbf{Set}}
\newcommand{\Cat}{\mathbf{Cat}}

\newcommand{\lc}[1]{  \mathsf{lc}^{#1} }
\newcommand{\rc}[1]{  \mathsf{rc}^{#1} }

\renewcommand{\l}{\mathsf{l}}

\newcommand{\premonoidal}[1]{ \left(#1, {\leftTens{}{}}, {\rightTens{}{}}, \tensu \right)}
\newcommand{\binoidal}[1]{(#1, {\leftTens{}{}}, {\rightTens{}{}}\!)}

\newcommand{\centreOf}[1]{\mathcal{Z}(#1)}

\newcommand{\leftTens}[2]{#1 \rtimes #2}
\newcommand{\rightTens}[2]{#1 \ltimes #2}
\newcommand{\ltie}{ \rtimes }
\newcommand{\rtie}{ \ltimes }

\newcommand{\Kl}{\mathcal{K}}

\newcommand{\freydCat}{\mathcal{F}}

\newcommand{\un}{\mathsf{u}}
\newcommand{\co}{\mathsf{c}}

\DeclareMathOperator{\bulletop}{\bullet}
\renewcommand{\vert}{\bulletop}
\newcommand{\vertsub}[1]{\vert}

  \newcommand{\Span}{\mathbf{Span}}

\newcommand{\modif}{m}

\newcommand{\twocellIso}[1]{\overset{\small{#1}}{\iso}}

\newenvironment{bprooftree}
  {\leavevmode\hbox\bgroup}
  {\DisplayProof\egroup}

\makeatletter
\renewcommand{\@marginparreset}{%
  \reset@font\footnotesize
  \footnote
  \raggedright
  \@setminipage
}
\makeatother


%% file: introduction.tex
A fundamental aspect of call-by-value functional programming languages is the distinction between \emph{values} and \emph{computations}. While
values are `pure' program fragments that can be passed around safely, 
computations may interact with their environment in the form of \emph{effects} (such
as raising exceptions, interacting with state, or behaving probabilistically), and must therefore be manipulated with care. 

Values and computations obey different algebraic properties, and in
particular only
computations are sensitive to the evaluation order.  
 For instance
	{\tt print \!"a"\!;\! print \!"b"}
is not equivalent to 
	{\tt print \!"b"\!;\! print \!"a"}. This is reflected in the
        denotational semantics of call-by-value languages, which
        consists of  a pair of categories: a monoidal category of
        values, and a \emph{premonoidal} category of
        computations. These are related by an identity-on-objects functor
coercing values into effect-free computations, and the
resulting structure is called a \emph{Freyd category} 
(\cite{Power1997env,Levy2003}).

In this paper we generalize these notions from categories to
bicategories.  The resulting theory includes models of
programming languages 
in which the morphisms are themselves objects 
with structure---spans, strategies, parameter spaces,
profunctors, open systems, \emph{etc.}---for which the notion of
composition uses a universal construction, such as a pullback or a pushout.
In these models, the 2-cells play a central role in characterizing the composition
operation for morphisms, and additionally provide refined semantic~information
	(see~\eg~\cite{Hilken1996,LICS2019, ong-tsukada,Olimpieri2021,
	  Kerinec2023}).

\subsection{Bicategorical models}
\label{sec:case-for-bicategories}

A bicategory is a 2-dimensional category in which the associativity and  unit laws
for the composition of morphisms are
replaced by invertible 2-cells satisfying coherence axioms \cite{Benabou1967}.  
Bicategories have recently found prominence as models of computational processes: see \eg~\cite{template-games,  Fiadeiro2007,Genovese2021, Baez2016}. 
We illustrate this with two simple examples: spans of sets, and graded
monads. For reasons of space we have omitted definitions of the basic notions in bicategory theory, such as pseudofunctors, pseudonatural
transformations, and modifications. For a textbook account, see~\eg~\cite{Benabou1967}.

\subparagraph{Bicategories of spans.}
The bicategory $\Span(\Set)$ has objects sets and 1-cells 
	$A \rightsquigarrow B$ 
spans of functions	
	$A \longleftarrow S \longrightarrow B$.
We can compose pairs of morphisms 
	$A \longleftarrow S \longrightarrow B$ 
and 
	$B \longleftarrow R \longrightarrow C$ 
        using a pullback in the category of sets, as on the left below:
	\begin{equation*}
          \begin{tikzcd}[column sep=.3em, row sep=0.3em]
	&[1em] & {R \circ S} \\
	& S && R \\
	A && B &&[1em] C
	\arrow[from=2-2, to=3-1]
	\arrow[from=2-2, to=3-3]
	\arrow[from=2-4, to=3-3]
	\arrow[from=2-4, to=3-5]
	\arrow[from=1-3, to=2-2]
	\arrow[from=1-3, to=2-4]
	\arrow["\lrcorner"{anchor=center, pos=0.125, rotate=-45}, draw=none, from=1-3, to=3-3]
\end{tikzcd}\qquad
	\begin{tikzcd}[column sep=0.7em, row sep=0.7em, scalenodes=0.6]
	&&& \bullet \\
	&& \bullet \\
	& \bullet && \bullet && \bullet \\
	\bullet && \bullet && \bullet && \bullet
	\arrow[from=3-2, to=4-3]
	\arrow[from=3-4, to=4-3]
	\arrow[from=3-4, to=4-5]
	\arrow[from=2-3, to=3-2]
	\arrow[from=2-3, to=3-4]
	\arrow["\lrcorner"{anchor=center, pos=0.125, rotate=-45}, draw=none, from=2-3, to=4-3]
	\arrow[from=3-6, to=4-5]
	\arrow[from=1-4, to=2-3]
	\arrow[from=1-4, to=3-6]
	\arrow[from=3-6, to=4-7]
	\arrow["\lrcorner"{anchor=center, pos=0.125, rotate=-45}, draw=none, from=1-4, to=3-4]
	\arrow[from=3-2, to=4-1]
      \end{tikzcd}
	\hspace{5mm}
	\begin{tikzcd}[column sep=0.7em, row sep=0.7em, scalenodes=0.6]
	&&& \bullet \\
	&&&& \bullet \\
	& \bullet && \bullet && \bullet \\
	\bullet && \bullet && \bullet && \bullet
	\arrow[from=3-2, to=4-3]
	\arrow[from=3-4, to=4-3]
	\arrow[from=3-4, to=4-5]
	\arrow[from=3-6, to=4-5]
	\arrow[from=3-6, to=4-7]
	\arrow[from=3-2, to=4-1]
	\arrow[from=1-4, to=3-2]
	\arrow[from=2-5, to=3-4]
	\arrow[from=2-5, to=3-6]
	\arrow[from=1-4, to=2-5]
	\arrow["\lrcorner"{anchor=center, pos=0.125, rotate=-45}, draw=none, from=1-4, to=3-4]
	\arrow["\lrcorner"{anchor=center, pos=0.125, rotate=-45}, draw=none, from=2-5, to=4-5]
	\end{tikzcd}
    \end{equation*}
This composition correctly captures a notion of `plugging together' spans, but is only 
associative in a weak sense, since the two ways of taking pullbacks
(on the right above) are not generally equal. But, by the universal property of
   pullbacks, they are canonically isomorphic as spans.%

\subparagraph{Kleisli bicategories for graded monads.}
For another example we consider monads graded by monoidal
categories. Formally, a graded monad on a category $\catC$ consists of a monoidal category $(\catE, \bullet, \tensu)$ 
of \emph{grades} and a lax monoidal functor 
	$T : \catE \to [\catC, \catC]$
(see~\eg~\cite{Smirnov2008,Mellies2012, Katsumata2014}).
In particular, this gives a functor $T_e : \catC \to \catC$ for every
$e \in \catE$, and natural transformations 
	$\mu_{e, e'}  : T_{e'} \circ T_e \To T_{e \bullet e'}$ and
	$\eta : \id \To T_{\tensu}$
corresponding to a multiplication and unit. 

 Previous Kleisli-like constructions for graded monads have used presheaf-enriched categories 
 	(\eg~\cite{Gaboardi2021,McDermottFlexible}),
 but there is also a natural bicategorical construction. 
The objects are those of $\catC$ and 1-cells $A \rightsquigarrow B$ consist of a grade $e$ and a map 
	$f : A \to T_eB$ in $\catC$.
The 2-cells $(e, f) \To (e', f')$  are re-gradings: maps $\gamma : e \to e'$ in $\catE$ such that
		$T_\gamma(B) \circ f = f'$.
The composition and identities use the multiplication and unit, as
for a Kleisli category. But, unless $\catE$ is strict monoidal, this
operation is only weakly associative and unital.

A concrete instance of this is the $\mathbf{coPara}$
construction on a monoidal
category $\catC$ (\cite{Fong2019, Cruttwell2022}), equivalently defined as the Kleisli bicategory for the
monad graded by $\catC$ itself and given by $T_C(A) = A \otimes C$.\\

The broader context for this
work is the recent occurrence of bicategories in the semantics of
programming languages. Bicategories of profunctors are
now prominent in the analysis of linear logic and the
$\lambda$-calculus (\cite{FioreSpecies, galal-profunctors,
  Kerinec2023}), and game semantics employs a variety of span-like
constructions that compose weakly (\cite{template-games,cg1}). These
models have also influenced the development of 2-dimensional type theories
(\cite{LICS2019, Olimpieri2021}).   This paper supports these
developments from the perspective of call-by-value languages. (The
connection to linear logic explains our insistence on monoidal rather
than cartesian Freyd bicategories.)

\subsection{Monoidal bicategories}

A monoidal bicategory is a bicategory equipped with a unit object and
a tensor product which is only weakly associative and unital. In the
categorical setting `weakly' typically means `up to isomorphism'; in
bicategory theory it typically means `up to \emph{equivalence}'.

\begin{definition}%
	\label{def:equivalence}
	An \emph{equivalence} between objects $A$ and $B$ in a
    bicategory $\B$ is a pair of 1-cells $f : A \to B$ and
        $\psinv{f} : B \to A$ together with invertible 1-cells
        $\un : \Id_A \To \psinv{f} \circ f$
        and
        $\co : f \circ \psinv{f} \To \Id_B$.
	 This is an \emph{adjoint equivalence} if the witnessing 2-cells $\un$ and $\co$ satisfy the usual triangle laws for an adjunction (see~\eg~\cite{Leinster2004}). 
  \end{definition}
  
  It is common in bicategory theory for definitions to ask for adjoint equivalences: these are easier to work with and no stronger than asking for just equivalences 
  	(see~\eg~\cite[Proposition 1.5.7]{Leinster2004}).  

The bicategorical version of a natural isomorphism is a
\emph{pseudonatural (adjoint) equivalence}: a pseudonatural transformation in which each 1-cell component has the structure of an (adjoint) equivalence. 
  
\begin{definition}[{\eg~\cite{Stay2016}}]
\label{def:monoidal-bicategory}
A \emph{monoidal bicategory} is a bicategory $\B$ equipped with a 
pseudofunctor $\otimes : \B \times \B \to \B$ and an object $I \in \B$, 
together with: 
\begin{itemize}
    \item 
    	pseudonatural adjoint equivalences $\alpha, \lambda$ and $\rho$ with 
    	components 
    		$\alpha_{A, B, C} : (A \otimes B) \otimes C \to 
    			{A \otimes (B \otimes C)}$
    		(the \emph{associator}),
    		$\lambda_A : I \otimes A \to A$,
    	and
    		$\rho_A : A \otimes I \to A$
    		(the \emph{unitors}); and
    \item 
    	invertible modifications
    		$\mathfrak{p}, \mathfrak{l}, \mathfrak{m}$ and $\mathfrak{r}$
    	with components as in \Cref{fig:monoidal-bicat-modifications},
    	subject to coherence axioms.
      \end{itemize}
  \end{definition}
  
  \begin{figure}
  	\vspace{-3mm}
      \centering
  	\input{diag-monoidal-bicategory-modifications}
  	\caption{The structural modifications of a monoidal bicategory}
  	\label{fig:monoidal-bicat-modifications}
  	\vspace{-2mm}
  \end{figure}
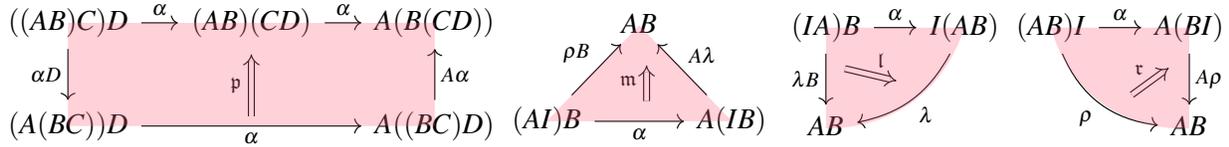

Monoidal bicategories have a technical algebraic definition but nonetheless
arise naturally.
For example, the cartesian product on the category
$\Set$ induces a monoidal structure on the bicategory $\Span(\Set)$.
Many other examples appear in a similar fashion: see~\cite{wester2019constructing}.

\paragraph{Coherence theorems.} 
\label{sec:coherence}

A careful reader might observe that the diagrams in \Cref{fig:monoidal-bicat-modifications} are not, strictly speaking, well-typed: for example, the anti-clockwise route around the diagram for $\pentagonator$ could denote 
	$(A\alpha \circ \alpha) \circ \alpha D$
or 
	$ A\alpha \circ (\alpha \circ \alpha D)$.
This is justified by a suitable \emph{coherence} theorem.

Typically, coherence theorems show that any two parallel 2-cells built out of the structural data are equal. 
Appropriate coherence theorems apply to bicategories~\cite{MacLane1985}, 
	pseudofunctors~\cite{Gurski2013}, 
	and
	(symmetric) monoidal bicategories%
		~(\cite{Gordon1995,Gurski2013,Gurski2013symmetricbicats}).
These results justify writing simply $\iso$ for composites of structural data in commutative diagrams of 2-cells, in much the same way as one does for monoidal categories.

As is common in the field, we rely heavily on the coherence of bicategories and pseudofunctors
when writing pasting diagrams of 2-cells. We omit all compositors and unitors for pseudofunctors, and  
ignore the weakness of 1-cell composition. 
Thus, even though our diagrams do not strictly type-check, coherence guarantees the resulting 2-cell is the same no matter how one fills in the structural details. 
For example, for a pseudofunctor $T$ on a monoidal bicategory we may write $T\mathfrak{l}$ as a 2-cell of type 
  $T(\lambda B) \To T\lambda \circ T\alpha$.
For a detailed justification see \eg%
	~\cite[Remark~4.5]{Gurski2013} or~\cite[\S2.2]{SchommerPries2009}.

\subsection{Premonoidal categories and Freyd categories}
\label{sec:strong-monads-intro}

Premonoidal categories generalize monoidal categories in that 
the tensor product $\tens$ is only functorial in each argument separately \cite{PowerRobinson1997}.
The lack of a monoidal ``interchange law'' reflects the fact that one cannot generally re-order the statements of an effectful program, even
if the data flow permits it. 
As a consequence, one can directly model effectful programs in a premonoidal category,
in the sense that a typed program
	$(\Gamma \vdash M : A)$ 
is modelled directly as an arrow
	$\Gamma \to A$ and
the result of substituting $M$ into another effectful program
	$(\Delta, x: A \vdash N : B)$ 
is modelled by the composite 
$
	 \Delta \otimes \Gamma 
	 	\xra{\Delta \otimes M}  
	 	\Delta \otimes A
	  	\xra{N} 
	  	B$. 
Thus, the composition of morphisms in a
premonoidal category should be understood as encoding control
flow. This is illustrated in Figure~\ref{fig:controlflow} using the
graphical calculus for premonoidal categories
(\cite{jeffrey1997premonoidal,DBLP:journals/corr/abs-2205-07664}),
where the dashed red line indicates control flow.
This direct interpretation contrasts with monadic approaches~(\cite{Moggi1989, Moggi1991}),
which rely on a monad whose structure may not be reflected in the syntax.

\input{fig-control-flow}

We can axiomatize the morphisms for which 
interchange does hold. Let $\catD$ be a category equipped with functors
	$A \ltie (-) : \catD \longrightarrow \catD$
and
	$(-) \rtie B : \catD \longrightarrow \catD$
for every $A, B \in \catD$, 
such that
	$A \ltie B = A \rtie B$. 
We write $A \tens B$, or
        just $AB$, for their joint value, and $\binoidal\catD$ is called a
        \emph{binoidal category}. 
A map  $f : A \to A'$ in $\catD$ is \emph{central} if the two diagrams
\begin{equation}
 \label{eq:centrality-cat}
	\begin{tikzcd}[column sep = 3 em]
		A B 
		\arrow{r}{A \ltie g}
		\arrow[swap]{d}{f \rtie B}
		&
		A B' 
		\arrow{d}{f \rtie B'}
		\\
		A'  B
		\arrow[swap]{r}{A' \ltie g}
		&
		A' B'
	\end{tikzcd}
	\hspace{16mm}
	\begin{tikzcd}[column sep = 3 em]
		B   A
		\arrow{r}{g \rtie A}
		\arrow[swap]{d}{B \ltie f}
		&
		B' A
		\arrow{d}{B' \ltie f}
		\\
		B  A'
		\arrow[swap]{r}{g \rtie A'}
		&
		B'  A'
	\end{tikzcd}
\end{equation}
commute for every $g :B \to B'$.  
Semantically, $f$ corresponds to a computation which may be run at any
point without changing the observable result. 

A premonoidal category is a binoidal category $\binoidal\catD$ with central structural isomorphisms $\alpha, \lambda$ and $\rho$ similar to those in a monoidal category. 
Unlike with monoidal categories, however, the associator $\alpha$ cannot be a natural transformation in all arguments simultaneously, because $\tens$ is not a functor on
$\catD$. Instead, we must ask for naturality in each argument separately, so the following three diagrams commute:
\begin{equation}
\label{eq:naturality-for-alpha}
\begin{tikzcd}[column sep=4em, row sep = 1.5em]
  (AB)C 
  \arrow{r}{(f \rtie B) \rtie C}  
  \arrow[swap]{d}[xshift=0mm]{\alpha}  &
  (A'B)C 
  \arrow{d}[xshift=0mm]{\alpha} \\
  A(BC) 
  \arrow{r}[swap]{f \rtie (B C)} & 
  A'(BC)
\end{tikzcd}
\hspace{4mm}
\begin{tikzcd}[column sep=4em, row sep = 1.5em]
  (AB)C \arrow{d}[swap]{\alpha} \arrow{r}{(A \ltie g) \rtie C} & 
  (AB')C \arrow{d}{\alpha}\ \\
  A(BC) \arrow{r}[swap]{A \ltie (g \rtie C)}  & A(B'C)
\end{tikzcd}
\hspace{4mm}
\begin{tikzcd}[column sep=4em, row sep = 1.5em]
  (AB)C \arrow{d}[swap]{\alpha}  \arrow{r}{(A B) \ltie h}&
  (AB)C'  \arrow{d}{\alpha}  \\
  A(BC') \arrow[]{r}[swap]{A \ltie (B \ltie h)} & A(BC')
\end{tikzcd}
\end{equation}

\begin{definition}[\cite{PowerRobinson1997}]
\label{def:premonoidal-category}
A \emph{premonoidal category} is a binoidal category 
	$\binoidal\catD$ 
equipped with a unit object $I$ and central isomorphisms 
	$\rho_A : A I \to A$, 
	$\lambda_A : IA \to A$ 
and 
	$\alpha_{A, B, C} : (AB)C \to A  (B  C)$ 
for every $A, B, C \in \catD$, 
natural in each argument separately
and satisfying the axioms for a monoidal category.
\end{definition}

One important contribution of this paper is to bicategorify the
notion of central morphism. We will
see that as we move  from categories to bicategories centrality evolves from property to structure
(\Cref{def:bino-bicat-centr}).

\paragraph{Freyd categories.} 
When modelling call-by-value languages in premonoidal categories, it
is natural to think of the values as effect-free
computations. Semantically, this is captured by Freyd
categories~\cite{Power1997env}, which are  premonoidal categories
together with a choice of effect-free maps.

Precisely, a Freyd category consists of a monoidal
category $\catV$ (often \emph{cartesian} monoidal), a premonoidal
category $\catC$, and an identity-on-objects functor $\J : \catV \to
\catC$ that strictly preserves the tensor product and structural
morphisms, and such that every morphism $\J(f)$ is central in
$\catC$. 

Although every premonoidal category $\catD$ canonically induces a Freyd
category $\centreOf{\catD} \hookrightarrow \catD$, where $\centreOf{\catD}$ is the
subcategory of central maps (called the \emph{centre}), there are several reasons to consider Freyd
categories directly. First, it does not always make sense to regard all
central maps as values: for instance, in a language with
commutative effects (\eg~probability), \emph{all} computations are central. Second, functors between binoidal categories do not in general preserve central maps, whereas morphisms of Freyd categories include a functor between the categories of values specifying how values are sent to values.

\paragraph{Relationship to monad models.}
Freyd categories encompass the
strong monad semantics of call-by-value proposed by
Moggi~(\cite{Moggi1989, Moggi1991}). 
Indeed, if 
	$(\catC, \otimes,I)$ 
is symmetric monoidal, then any strength for a monad $(T, \mu, \eta)$ on $\catC$ induces a
premonoidal structure on the Kleisli category $\catC_T$, and $\eta
\circ (-) : \catC \to \catC_T$ becomes a Freyd category. Conversely,
a Freyd category corresponds to a monad whenever $\J$ has a right
adjoint \cite{PowerRobinson1997}. 
This adjoint is necessary if the programming language has
higher-order functions, but some `first-order' 
Freyd categories are not known to arise from a monad
	(\eg~\cite{Thielecke1997,Power2002, Staton2017}).

\subsection{Contributions and outline}

The central aim of this paper is to introduce definitions of 
	premonoidal bicategories (\Cref{def:premonoidal-bicategory}) 
and
	Freyd bicategories (\Cref{def:freyd-bicategory}).
Premonoidal structure relies on an adequate
notion of centrality for 1-cells and 2-cells in a bicategory (\Cref{def:bino-bicat-centr}). Freyd bicategories
then require a coherent assignment of centrality data, which leads to
subtle compatibility issues, outlined in  
\Cref{sec:premonoidal-bicategories} and \Cref{sec:freyd-bicategories}.
	
As ever with bicategorical definitions
 	(see~\eg~\cite[\S2.1]{SchommerPries2009}), 
the main difficulty is in ensuring the right axioms on the 2-cells.
We therefore give further justification for our definitions. 
On the one hand, we show that our definitions are not too strict: they capture natural examples, 
presented in \Cref{sec:first-examples-of-premonoidal-bicats} and \Cref{sec:examples-of-Freyd-bicategories}.
On the other hand, we show that our definitions are not too weak: the well-known correspondence between Freyd categories and actions~\cite{LevyBook} lifts to our setting
	(\Cref{sec:correspondence-theorem}).
We note that our definition of action is extracted from standard
higher-categorical constructions, and so our work connects to an already-existing and well-understood body of theory.

\vspace{3mm}

	The definition of premonoidal bicategory presented here is based on that in the ArXiv preprint~\cite{DBLP:journals/corr/abs-2304-11014}. 
	For reasons of space, we sketch only the proof of the main theorem 		
		\Cref{res:correspondence-theorem}
	here. 
	For more proofs, see the longer version of this paper, available on the authors' webpages.

%% file: diag-monoidal-bicategory-modifications.tex
\hspace{-6mm}
\begin{tikzcd}
	[column sep=1.3em, 
		execute at end picture={
					\foreach \nom in  {A,B,C, D,E}
		  				{\coordinate (\nom) at (\nom.center);}
					\fill[\monoidalpcolour,opacity=\opacity] 
		  				(A) -- (B) -- (C) -- (D) -- (E);
		}
		]
	\alias{A}
	{\left((A  B)  C\right)}  D 
	\arrow{r}{\alpha}
	\arrow[swap]{d}{\alpha D}
	\arrow[	"{\mathfrak{p}\:}"{left},
					from=2-2, 
					to=1-2,  
					Rightarrow, 
					shorten = 2pt
					]
	&
	\alias{B}
	(AB) (CD)
	\arrow{r}{\alpha}
	&
	\alias{C}
	A{\left(B(C D) \right)}
	\\
	\alias{E}
	{\left( A (BC) \right)}D
	\arrow[swap]{rr}{\alpha}
	&
	\:
	&
	\alias{D}
	A{\left((BC) D\right)}
	\arrow[swap]{u}{A\alpha}
\end{tikzcd}
\hspace{-3mm}
\begin{tikzcd}[
	column sep = .5em,
	execute at end picture={
								\foreach \nom in  {A,B,C}
					  				{\coordinate (\nom) at (\nom.center);}
								\fill[\monoidalmcolour,opacity=\opacity] 
					  				(A) -- (B) -- (C);
	  }
	]
\: &
\alias{C}
{A  B}
&
\: 
\\
\alias{A}
{(A  I) B} 
\arrow{ur}{\rho B}
\arrow[swap]{rr}{\alpha}
&
\:
& 
\alias{B}
{A (I B)} 
\arrow[name=0]{ul}[swap]{A\lambda}
\arrow["{\mathfrak{m}\:}"{}, 
 				shift right=1, 
 				shorten <=6pt, 
 				shorten >=8pt, 
 				Rightarrow, 
 				from=2-2, 
 				to=1-2]
\end{tikzcd}
\hspace{-2mm}
\begin{tikzcd}[
	column sep=1.5em,
	execute at end picture={
			\foreach \nom in  {A,B,C}
  				{\coordinate (\nom) at (\nom.center);}
			\fill[\monoidallcolour,opacity=\opacity] 
  				(A) -- (B) to[bend left = 36] (C);
		}
	]
\alias{A}
{(I A)  B} 
& 
\alias{B}
{I  (A  B)} \\
\alias{C}
{A B}
\arrow["\alpha", from=1-1, to=1-2]
\arrow[""{name=0, anchor=center, inner sep=0}, bend left = 25, "\lambda", from=1-2, to=2-1]
\arrow[""{name=1, anchor=center, inner sep=0}, "{\lambda  B}"', from=1-1, to=2-1]
\arrow["{\mathfrak{l}}"{}, 
	    					shift left=1,
	    					yshift = 1mm,
	    					shorten <=6pt, 
	    					shorten >=6pt, 
	    					Rightarrow, 
	    					from=1, 
	    					to=0]
\end{tikzcd}		
\hspace{-3mm}
\begin{tikzcd}[
	column sep = 1.5em, 
	execute at end picture={
											\foreach \nom in  {A,B,C}
								  				{\coordinate (\nom) at (\nom.center);}
											\fill[\monoidalrcolour,opacity=\opacity] 
								  				(A) to[bend right=38] (B) -- (C);
				  }
]
\alias{A}
(AB)I 
\arrow{r}{\alpha}
\arrow[bend right =30, name=0]{dr}[swap]{\rho}
&
\alias{C}
A(BI)
\arrow{d}{A\rho}
\\
\:
\arrow["{\mathfrak{r}}",
			Rightarrow, 
			xshift=6mm,
			yshift=1mm,
			shorten <=14pt, 
			shorten >=14pt, 
			from=2-1, 
			to = 1-2]
&
\alias{B}
AB
\end{tikzcd}

%% file: fig-control-flow.tex
 \begin{wrapfigure}[16]{r}[0pt]{5cm}
\vspace{-2mm}
\begin{framed}
\centering
  \begin{minipage}{\linewidth}
  \centering
\scalebox{.95}{  \begin{tikzpicture}
\node[draw] (M) at (0, 2) {$M$};
\node[draw] (N) at (1, 1.2) {$N$};
\node[draw] (P) at (0.5, 0.2) {$\ \ \ \ P \ \ \ \ $};
\node (gam) at (0, 2.7) {$\Gamma$};
\draw[thick, color=red, dash pattern=on 3pt off 1pt] (0.2, 2.7) -- (0.2, 2.25);
\draw[thick, color=red, dash pattern=on 3pt off 1pt] (0.2, 1.75) to[out=270, in=90]
(0.85, 1.45);
\draw[thick, color=red, dash pattern=on 3pt off 1pt] (0.88, 0.95) to[out=270, in=90]
(0.88, 0.45);
\draw[thick, -latex, color=red, dash pattern=on 3pt off 1pt] (0.88, -0.05) to[out=270, in=90]
(0.88, -0.3);

\node (del) at (1, 2.7) {$\Delta$};
\node (C) at (0.5, -0.4) {$C$};

\draw (M) -- (0, 0.45);
\draw (N) -- (1, 0.45);
\draw (P) -- (C);
\draw (gam) -- (M);
\draw (del) -- (N);
\end{tikzpicture}
\ \  
\raisebox{1.5cm}{$\neq$}
\ \ 
\begin{tikzpicture}
\node[draw] (M) at (0, 1.2) {$M$};
\node[draw] (N) at (1, 2) {$N$};
\node[draw] (P) at (0.5, 0.2) {$\ \ \ \ P \ \ \ \ $};
\node (gam) at (0, 2.7) {$\Gamma$};
\draw[thick, color=red, dash pattern=on 3pt off 1pt] (0.8, 2.7) -- (0.8, 2.25);
\draw[thick, color=red, dash pattern=on 3pt off 1pt] (0.8, 1.75) to[out=270, in=90]
(0.2, 1.45);
\draw[thick, color=red, dash pattern=on 3pt off 1pt] (0.2, 0.95) to[out=270, in=90]
(0.2, 0.45);
\draw[thick, -latex, color=red, dash pattern=on 3pt off 1pt] (0.2, -0.05) to[out=270, in=90]
(0.2, -0.3);

\node (del) at (1, 2.7) {$\Delta$};
\node (C) at (0.5, -0.4) {$C$};

\draw (M) -- (0, 0.45);
\draw (N) -- (1, 0.45);
\draw (P) -- (C);
\draw (gam) -- (M);
\draw (del) -- (N);
\end{tikzpicture}}
\end{minipage}%
\\[-1mm]
\begin{minipage}{\linewidth}
\[   \begin{tikzcd}[column sep=0em, row sep=1em]
	& {\Gamma \otimes \Delta} \\
	{A \otimes \Delta} && {\Gamma \otimes B} \\
	& {A \otimes B} \\
	& C
	\arrow["{M \otimes \Delta}"', from=1-2, to=2-1]
	\arrow["{A \otimes N}"', from=2-1, to=3-2]
	\arrow["P", from=3-2, to=4-2]
	\arrow["{\Gamma \otimes N}", from=1-2, to=2-3]
	\arrow["{M \otimes B}", from=2-3, to=3-2]
	\arrow["{\neq}"{description}, draw=none, from=2-3, to=2-1]
      \end{tikzcd}
    \]
  \end{minipage}
\vspace{-2mm}
\end{framed}
\vspace{-3mm}
\captionsetup{justification=centering}
\caption{Failure of interchange in a premonoidal category. 
  }
\label{fig:controlflow}
\vspace{0mm}
 \end{wrapfigure}

%% file: premonoidal.tex
Just as in the categorical setting (\eg~\cite{PowerRobinson1997}), our starting point is \emph{binoidal} structure.

\begin{definition}
A \emph{binoidal bicategory} $\binoidal{\baseCat}$ is a bicategory $\baseCat$ with pseudofunctors 	
	$A \ltie (-)$
and
	$(-) \rtie B$
for every $A, B \in \baseCat$, such that $A \ltie B = A \rtie B$. 
We write $A \tens B$,  or just $AB$, for the joint value on objects.
\end{definition}

As is standard when moving from categories to bicategories, the category-theoretic property of centrality becomes extra structure in a binoidal bicategory. For the definition, we observe that the diagrams defining centrality (\ref{eq:centrality-cat}) amount to requiring that $f$ induces two natural transformations:
\begin{equation}
\label{eq:transformation-for-centrality}
\begin{aligned}
	\lc{f} : A \ltie (-) \To A' \ltie (-) &
	\qquad
	,
	\qquad
	\lc{f}_B :=  \big( A \ltie B = A \rtie B \xra{f \rtie B} A' \rtie B = A' \ltie B \big)  \\
	\rc{f} : (-) \rtie A \To (-) \rtie A' &
	\qquad
	,
	\qquad
	\rc{f}_B := \big( B \rtie A = B \ltie A \xra{B \ltie f} B \ltie A' = B \rtie A' \big)
\end{aligned}
\end{equation}
This lifts naturally to the bicategorical setting, and gives an immediate notion of centrality for 2-cells.

\vspace{2mm}
\noindent
\begin{minipage}{0.7\textwidth}
\begin{definition}
  \label{def:bino-bicat-centr}
  Let $\binoidal{\B}$ be a binoidal bicategory.
A 
	\emph{central 1-cell} 
is a 1-cell $f : A \to A'$ equipped with invertible 2-cells as on the right
for every $g: B \to B'$, such that the 1-cells in (\ref{eq:transformation-for-centrality}) are the components of pseudonatural transformations 
	$\lc{f} : A \ltie (-) \To A' \ltie (-)$
and 
	 $\rc{f} : (-) \rtie A \To (-) \rtie A'$.
A \emph{central 2-cell} $\sigma$ between central 1-cells
	$(f, \lc{f}, \rc{f})$ 
and 
	$(f', \lc{f'}, \rc{f'})$
  	is a 2-cell $\sigma : f \To f'$
such that the 2-cells $\sigma \rtie B$ and $B \ltie \sigma$
(for $B \in \baseCat$) define modifications
        $\lc{f} \To \lc{f'}$ 
and
        $\rc{f} \To \rc{f'}$, 
respectively.
\end{definition}
\end{minipage}
\hfill
\begin{minipage}{0.3\textwidth}
	\centering
	\begin{tikzcd}[
			row sep = 1em,
			column sep = 3em
		]
		A B 
		\arrow{r}{A \ltie g}
		\arrow[swap]{d}{f \rtie B}
		&
		A  B' 
		\arrow{d}{f \rtie B'}
		\arrow[draw=none, yshift=-8mm, swap, shorten=.5em]{l}{\lc{f}_{g}} %
		\\
		A'  B
		\arrow[swap]{r}{A' \ltie g}
		&
		A' B'
	\end{tikzcd}
	\begin{tikzcd}[
				row sep = 1em,
				column sep = 3em
			]
		B   A
		\arrow{r}{g \rtie A}
		\arrow[swap]{d}{B \ltie f}
		&
		B'  A
		\arrow{d}{B' \ltie f}
		\arrow[draw=none, yshift=-8mm, swap, shorten=.5em]{l}{\rc{f}_{g}} 
		\\
		B  A'
		\arrow[swap]{r}{g \rtie A'}
		&
		B' A'
\end{tikzcd}              
\end{minipage}
\vspace{2mm}

Every monoidal bicategory $\monoidal{\baseCat}$ has a canonical
binoidal structure, with  $\ltie$ and $\rtie$ directly
induced from the monoidal structure by fixing one argument. 
Every 1-cell $f$ in $\B$ is canonically central, 
with $\lc{f}_g$ given by the interchange isomorphism induced by the pseudofunctor structure of $\tens$, and  $\rc{f}_g$ by $(\lc{g}_f)^{-1}$:
\begin{equation}
	\label{eq:centrality-from-interchange}
		\lc{f}_g :=
		\big(
			(f \tens B') \circ (A \tens g) 
			\XRA{\iso}
			(f \tens g)
			\XRA{\iso}
			(A' \tens g) \circ (f \tens B)
		\big).
\end{equation}
By the functoriality of $\otimes$, every 2-cell is
central with respect to this structure.

We will define premonoidal bicategories as binoidal bicategories with
central structural equivalences. As in \Cref{def:premonoidal-category}, the
associator $\alpha$ for the tensor product can only be 
pseudonatural in each argument separately, because $\tens$ is
not a functor of two arguments.  
We therefore need a family of equivalences
	$\alpha_{A, B, C} : (A \tens B) \tens C \to A \tens (B \tens C)$ 
together with invertible 2-cells 
	$\cellOf{\alpha}_{f, B, C},
		\cellOf{\alpha}_{A, g, C}$
and 
	$\cellOf{\alpha}_{A, B, h}$
filling the three squares in 
	(\ref{eq:2-cells-for- alpha}),
so that we get three families of pseudonatural transformations:
\begin{align}
\label{eq:2-cells-for- alpha}
  (\alpha_{-, B, C}, \cellOf{\alpha}_{-, B, C})& : (- \rtie B) \rtie C \To (-) \rtie (B \tens C) \nonumber \\
  (\alpha_{A, -, C}, \cellOf{\alpha}_{A, -, C}) &: (A \ltie -) \rtie C \To A \ltie (- \rtie C) \\
  (\alpha_{A, B, -}, \cellOf{\alpha}_{A, B, -}) &: (A \tens B) \ltie (-) \To A \ltie (B \ltie -) \nonumber 
\end{align}

A premonoidal bicategory also involves structural modifications
corresponding to those of \Cref{fig:monoidal-bicat-modifications}.
Here the 2-dimensional structure introduces new subtleties. 
For example, one side of  modification
$\mathfrak{l}$ in \Cref{fig:monoidal-bicat-modifications}
uses the pseudonatural transformation with components 
	$\lambda_A \tens B  : (\tensu A)  B \to A B$. 
For $g : B \to B'$, the 2-cell witnessing pseudonaturality of this
transformation is the canonical isomorphism that interchanges $\lambda_A$ and $g$.
This 2-cell does not exist in a premonoidal bicategory, so instead we
must use the centrality witness $\lc{\lambda_A}_g$ for
$\lambda_A$. 
Thus, we define $\mathfrak{l}$ to be a family of 2-cells 
	$\mathfrak{l}_{A, B} : (\lambda_A \rtie B) \To \lambda_{A
          \tens B} \circ \alpha_{\tensu, A, B}$, pictured on the left below,
       inducing modifications in $\homBicat{\B}{\B}$ of both types
        on the right below:
\begin{equation*}
	\begin{tikzcd}[column sep = .8em, row sep = 1em]
	 (\tensu  A)  B 
	 \arrow{rr}{\lambda \rtie B}
	 \arrow[bend right = 12, swap]{dr}{\alpha_{\tensu, A, B}}
	 &
	 \:
	 \arrow[phantom]{d}{  \mathfrak{l}_{A, B} }
	 &
	 A B
	 \\
	 \: 
	 &
	 \tensu (A  B)
	 \arrow[bend right = 12, swap]{ur}{\lambda_{A  \tens B}}
	 \: 
       \end{tikzcd}
	\qquad
  \begin{tikzcd}[column sep = -1em, row sep = 1em]
	 (\tensu \ltie -) \ltie B 
	 \arrow{rr}{\lambda \rtie B}
	 \arrow[bend right = 12, swap]{dr}{\alpha_{\tensu, -, B}}
	 &
	 \:
	 \arrow[phantom]{d}{  \mathfrak{l}_{-, B} }
	 &
	 (- \rtie B)
	 \\
	 \: 
	 &
	 \tensu \ltie (- \rtie B)
	 \arrow[bend right = 12, swap]{ur}{\lambda_{- \rtie B}}
	 \: 
	\end{tikzcd}
        \qquad
	\begin{tikzcd}[column sep = -1em, row sep = 1em]
	(\tensu A) \ltie (-)
	\arrow{rr}{  \lc{\lambda} }
	\arrow[bend right = 12, swap]{dr}{\alpha_{\tensu, A, -}}
	&
	\:
	\arrow[phantom]{d}{  \mathfrak{l}_{A, -} }
	&
	(A \ltie -)
	\\
	\: 
	&
	\tensu \ltie (A \rtie -)
	\arrow[bend right = 12, swap]{ur}{\lambda_{A \ltie -}}
	\: 
	\end{tikzcd}
\end{equation*}
Notice that the middle diagram appears exactly as in the definition of
a monoidal bicategory; no adjustments are necessary because each
transformation is pseudonatural in the open argument without any
assumptions of centrality.

Modulo the subtleties just outlined, our main definition is a natural extension of the categorical one. We abuse notation by saying ``$f$ is central'' to mean $f$ comes with chosen $\lc{f}$ and $\rc{f}$ making $(f, \lc{f}, \rc{f})$ a central 1-cell and saying
	``the pseudonatural transformation $\eta$ is central'' 
to mean each 1-cell component $\eta_A$ is central. 
\begin{definition}
	\label{def:premonoidal-bicategory}
  A \emph{premonoidal bicategory} is a binoidal bicategory
  $\binoidal{\baseCat}$ equipped with a unit object $\tensu \in \B$,
  together with the following data:
	\begin{enumerate}
	\item 	
		For every $A \in \baseCat$, central pseudonatural adjoint equivalences
			$\lambda_A : \tensu \ltie A \to A$
		and
			$\rho_A : A \rtie \tensu \to A$;
        \item  
        	For every $A, B, C \in \baseCat$, an adjoint equivalence 
        		$\alpha_{A, B, C} : (A \tens B)  \tens C \to A \tens (B \tens C)$
        	with 2-cells as in \eqref{eq:2-cells-for- alpha} inducing central pseudonatural equivalences in each component separately;
	\item 
		For each $A, B, C, D \in \baseCat$, invertible central 2-cells
			$\mathfrak{p}_{A, B, C, D}, \mathfrak{m}_{A, B}, \mathfrak{l}_{A, B}$ 
		and $\mathfrak{r}_{A, B}$, forming
                modifications in each argument as in
          		\Cref{fig:premonoidal-bicat-modifications} 
          	or, if not shown there, as in a monoidal bicategory.
	\end{enumerate}
	This data is subject to the same equations between 2-cells as in a monoidal bicategory. %
\end{definition}

Note that we cannot ask for the 2-cell components of the structural transformations to be central: for example, $\cellOf{\rho}_f$ has type 
	$\rho_{A'} \circ (f \rtie \tensu) \To f \circ \rho_A$, 
but $f$ may not be a central map. 
Also note that, although we have changed the conditions for 
$\mathfrak{p}, \mathfrak{m}, \mathfrak{l}$ and $\mathfrak{r}$
to be modifications, their type as 2-cells has not changed, and thus
the equations for a monoidal bicategory are still well-typed.

\begin{figure*}
	\input{premonoidal-bicat-modification-diags-figure}  %
	\caption{
		Modification axioms for the structural 2-cells of a premonoidal bicategory, 
			where they differ from those of a monoidal
                        bicategory. (To save space we suppress $\ltie$ and $\rtie$: these can be inferred.)
	}
	\label{fig:premonoidal-bicat-modifications}
\end{figure*}

Just as every premonoidal category has a centre, so does every premonoidal bicategory. 

\begin{definition} 
	For a premonoidal category $\premonoidal{\baseCat}$, denote by
	$\centreOf \baseCat$ the bicategory with the same objects, whose
	1-cells and 2-cells are the central 1-cells and central 2-cells in
	$\baseCat$. 
	Composition is defined using composition in $\homBicat{\B}{\B}$, and the identity on $A$ is $\Id_A$ with the identity transformations.
\end{definition}

The pseudofunctors $A \ltie (-)$ and $(-) \rtie B$ lift to the centre. Because 
	$A \ltie (-)$ 
is  a pseudofunctor, then for any central 1-cell
	$(f, \lc{f}, \rc{f})$
we already have pseudonatural transformations
	$A \ltie \lc{f}$
and 
	$A \ltie \rc{f}$
in $\B$.
To disambiguate between these transformations and the action of $A \ltie (-)$ on central 1-cells, we denote the latter by
	$A \ltie (f, \lc{f}, \rc{f})
		:= (A \ltie f, \lc{A \ltie f}, \rc{A \ltie f})$, 
and likewise for $(-) \rtie B$.

\begin{proposition}
  \label{res:assoc-binoidal-bicat-has-binoidal-centre}
  Let $\premonoidal{\baseCat}$ be a premonoidal bicategory. 
  	For every $A, B \in \baseCat$ the operations  
  		$A \ltie (-)$
  	and 
  		$(-) \rtie B$
  	induce pseudofunctors on~$\centreOf\B$.
\end{proposition}
\begin{proof}[Proof sketch]
We only sketch the action of $X \ltie (-)$ and $(-) \rtie X$ on a central 1-cell 
$(f, \lc{f}, \rc{f}) : A \to A'$. For $g : B \to B'$, the 2-cell
$\lc{X\ltie f}_g$ is uniquely determined by the equation
 	\begin{equation}
  		\label{eq:diag-bowtie-on-centre}
  		\input{diag-bowtie-on-centre}
  		\vspace{-3mm}
              \end{equation}
in which, for clarity, we have omitted object names and left implicit the functors
$\ltie, \rtie$, which can be inferred. This is a valid definition for
$\lc{X \ltie f}_g$ because $\alpha$ is an equivalence, and all 2-cells involved are
invertible.
We similarly construct 2-cells $\lc{f \rtie X}_g$, $\rc{X \ltie f}_g$,
$\rc{f \rtie X}_g$. The rest of the proof consists of routine
verifications.

\end{proof}

\subsection{Examples of premonoidal bicategories}
\label{sec:first-examples-of-premonoidal-bicats}

\paragraph{State-passing style.}
Power \& Robinson motivate their definition of premonoidal categories by considering an uncurried version of the State monad~\cite{PowerRobinson1997}: for a symmetric monoidal category $\monoidal\catC$ and an object $S \in \catC$ modelling a set of states, one can model a program from $A$ to $B$ interacting with the state as a morphism 	
	$S \tens A \to S \tens B$.
The same applies bicategorically.

\begin{restatable}
	[{\cf~\cite{PowerRobinson1997}, \cite[Example~A.1]{Levy2003}}]
	{lemma}
	{StatePassingStylePremonoidal}
\label{res:premonoidal-structure-from-writer-monad}
Let 
	$\monoidal{\B}$
be a symmetric monoidal bicategory (\eg~\cite{Stay2016})
and $S \in \B$. Define a bicategory $\K$ with the same objects as $\B$, hom-categories 
	$\K(A, B) := \B(S \tens A, S \tens B)$,
and composition and identities as in $\B$. 
Then $\K$ admits a canonical premonoidal structure.
\end{restatable}

For the binoidal structure, one whiskers with the canonical pseudonatural equivalences:
\begin{equation}
\begin{aligned}
\label{eq:binoidal-structure-for-State}
	f \rtie B &:= 
			\big(
				S(AB)
				\xra{\simeq}
				(SA)B
				\xra{f \tens B}
				(SA')B
				\xra{\simeq}
				S(A'B)
			\big) \\
	A \ltie g &:= 
		\big(
			S(AB)
			\xra{\simeq}
			A(SB)
			\xra{A \tens g}
			A(SB')
			\xra{\simeq}
			S(AB')
		\big) 
\end{aligned}
\end{equation}
The structural transformations are then given by composing the structural transformations in $\B$ with the naturality 2-cells for the equivalences in (\ref{eq:binoidal-structure-for-State}).

\paragraph*{Bistrong graded monads.}

It is well-known that if a monad $T$ on a monoidal category $\monoidal\catC$ is \emph{bistrong}, meaning that it is equipped with a left strength 
	$t_{A,B} : A \tens TB \to T(A \tens B)$
and a right strength
	$s_{A, B} : T(A) \tens B \to T(A \tens B)$,
and these strengths are compatible in the sense that the two canonical maps 
	$(A \tens T(B)) \tens C \to T{\big(A \tens  (B \tens C)\big)}$
are equal,
then $\catC_T$ is premonoidal 
(see~\eg~\cite{McDermott2022}).
(This definition is obscured in the symmetric setting, because if $\catC$ is symmetric every strong monad is canonically bistrong.)
A similar fact applies to the Kleisli bicategory $\Kl_T$ for a graded monad defined in \Cref{sec:case-for-bicategories}.
To state this we need to define bistrong graded monads: we make a small adjustment to Katsumata's definition of strong graded monads
	\cite[Definition~2.5]{Katsumata2014}.
An endofunctor $T : \catC\to \catC$ equipped with two strengths 
	$t$ %
and 
	$s$ %
which are compatible in the sense above is called  \emph{bistrong} (see~\eg~\cite{McDermott2022}).
\begin{definition}%
	A \emph{bistrong graded monad} on a monoidal category $\monoidal\catC$ consists of a monoidal category 
		$(\catE, \bullet, \tensu)$ 
	of grades and a lax monoidal functor 
		$T : \catE \to [\catC, \catC]_{\text{bistrong}}$, 
	where
		$[\catC, \catC]_{\text{bistrong}}$
	is the category of bistrong endofunctors and natural transformations that commute with both strengths (see~\eg~\cite{McDermott2022}).
\end{definition}

Thus, a bistrong graded monad is a graded monad equipped with natural transformations
	$t_{A, B}^e :  A \tens T_e(B) \to T_e(A \tens B)$
and
	$s_{A, B}^e : T_e(A) \tens B \to T_e(A \tens B)$
for every grade $e$, compatible with the graded monad structure and with maps between grades.
\label{res:l-r-on-graded-monad-bicat}
One  then obtains strict pseudofunctors 
	$A \ltie (-), (-) \rtie B : \Kl_T \to \Kl_T$
for every $A, B \in \Kl_T$, defined similarly to the premonoidal structure on a Kleisli category:
\begin{equation*}
\begin{aligned}
	A \ltie g \:\:&=\:\: 
		\big(
			AB 
				\xra{A \tens g} 
			AT_e(B')
				\xra{t^e}
			T_e(AB')
		\big)	 
\quad 
,
\quad
	f \rtie B \:\:&=\:\: 
		\big(
			AB 
				\xra{f \tens B}
			T_e(A')B
				\xra{s^e}
			T_e(A'B)
		\big).
\end{aligned}
\end{equation*}
Moreover, every $f \in \catC(A, A')$  determines a `pure' 1-cell
in $\Kl_T$, as
	$
		\para{f} := \big(A \xra{f} A' {\xra{\eta_{A'}} T_I A'} \big).
	$
This 1-cell canonically determines a central 1-cell, with 
	$\lc{\para{f}}$ 
and 
	$\rc{\para{f}}$
given by the canonical isomorphism in $\catC$; 
in particular, 
	$\lc{\tilde{f}}_{\tilde{g}} = (\rc{\tilde{g}}_{\tilde{f}})^{-1}$
for every $g \in \catC(B, B')$.
The structural transformations are then all of the form $\widetilde{\sigma}$ for $\sigma$ a structural transformation in $\catE$, and the structural modifications are all canonical isomorphisms of the form $\tensu^{\tens i} \xra\iso \tensu^{\tens j}$ for $i, j \in \Nat$. 
Summarizing, we have the following. 

\begin{restatable}{proposition}{GradedMonadKleisliIsPremonoidal}
Let $(T, \mu, \eta)$ be a bistrong graded monad on $\monoidal\catC$ 
	with grades $(\catE, \bullet, \tensu)$.
Then the bicategory $\Kl_T$ has a canonical choice of premonoidal structure. 

\end{restatable}

\paragraph{\emph{Un}natural transformations.}
For any category $\catC$ the category 
	$[\catC, \catC]_{u}$
of functors and \emph{un}natural transformations 
	(\ie~families of maps $\sigma_C: {FC \to GC}$ with no further conditions) 
is strictly premonoidal.  
This is almost by definition, because  Power \& Robinson define a strict premonoidal category to be a monoid with respect to the funny tensor product $\tens$ on the category $\Cat$~\cite{PowerRobinson1997}. 
A version holds bicategorically.

\begin{restatable}[{\cf~\cite{PowerRobinson1997}}]
	{lemma}
	{PremonoidalStructureOnUnnatTrans}
For any bicategory $\B$, let 
	$[\B, \B]_u$
denote the bicategory with objects pseudofunctors $F : \B \to \B$, 1-cells 
	$F \to G$
families of maps 
	$\{ \sigma_B : FB \to GB \st B \in \B \}$,
and 2-cells $\sigma \To \tau$ families of 2-cells
	$\{ \modif_B : \sigma_B \To \tau_B  \st B \in \B \}$.
Then 
	$[\B, \B]_u$
admits a premonoidal structure given by composition.
\end{restatable}

%% file: premonoidal-bicat-modification-diags-figure.tex
	\begin{center}
	\vspace{-3mm}
	\begin{tikzcd}[scalenodes=0.9, column sep=1em, row sep=.9em]
			\: 
			&
			(-B) (CD)
			\arrow{dr}{\alpha_{-, B, CD}}
			&
			\:
			\\
			{\left((-  B)  C\right)}  D 
			\arrow{ur}{\alpha_{-B, C, D}}
			\arrow[swap]{d}{\alpha_{-, B, C} D}
			\arrow[	"{\mathfrak{p}_{-, B, C, D}\:}"{left},
							xshift=6mm,
							from=3-2, 
							to=1-2,  
							Rightarrow, 
							shorten <= 10pt, 
							shorten >= 10pt
							]
			&
			\:
			&
			(-){\left(B(C D) \right)}
			\\
			{\left( - (BC) \right)}D
			\arrow[swap]{rr}{\alpha_{-, BC, D}}
			&
			\:
			&
			(-){\left((BC) D\right)}
			\arrow[swap]{u}{\rc{\alpha}}
		\end{tikzcd}
		\hspace{7mm}
	\begin{tikzcd}[scalenodes=0.9, column sep=1em, row sep=.9em]
			\: 
			&
			(AB) (C-)
			\arrow{dr}{\alpha_{A, B, C-}}
			&
			\:
			\\
			{\left((A  B)  C\right)}  (-) 
			\arrow{ur}{\alpha_{AB, C, -}}
			\arrow[swap]{d}{\lc{\alpha}}
			\arrow[	"{\mathfrak{p}_{A, B, C, -}\:}"{left},
							xshift=6mm,
							from=3-2, 
							to=1-2,  
							Rightarrow, 
							shorten <= 10pt, 
							shorten >= 10pt
							]
			&
			\:
			&
			A{\left(B(C -) \right)}
			\\
			{\left( A (BC) \right)}(-)
			\arrow[swap]{rr}{\alpha_{A, BC, -}}
			&
			\:
			&
			A{\left((BC) (-)\right)}
			\arrow[swap]{u}{A\alpha_{B, C, -}}
		\end{tikzcd}		
		\end{center}
		\vspace{-3mm}
		\begin{center}
		\hspace{-3mm}
		\begin{tikzcd}[scalenodes=0.9, column sep = 0em, row sep = .8em]
				\: &
				 {(-)  B}
				 &
				 \: 
				 \\
				{(-  I) B} 
				\arrow[bend left]{ur}{\rho B}
				\arrow[swap]{rr}{\alpha_{-, \tensu, B}}
				&
				\:
				& 
				{(-) (I B)} 
				\arrow[name=0, bend right]{ul}[swap]{\rc{\lambda}}
				 \arrow["{\mathfrak{m}_{-, B}\:}"{}, 
				 				shift right=1, 
				 				yshift=1mm,
				 				shorten <=2pt, 
				 				shorten >=4pt, 
				 				Rightarrow, 
				 				from=2-2, 
				 				to=1-2]
		\end{tikzcd}
		\begin{tikzcd}[scalenodes=0.9, column sep = 0em, row sep = .9em]
				\: &
				 {A (-)}
				 &
				 \: 
				 \\
				{(A  I) (-)} 
				\arrow[bend left]{ur}{\lc{\rho}}
				\arrow[swap]{rr}{\alpha_{A, I, -}}
				&
				\:
				& 
				{A (I -)} 
				\arrow[name=0, bend right]{ul}[swap]{A\lambda}
				 \arrow["{\mathfrak{m}_{A,-}\:}"{}, 
				 				shift right=1, 
				 				shorten <=4pt, 
				 				shorten >=2pt, 
				 				Rightarrow, 
				 				from=2-2, 
				 				to=1-2]
				\end{tikzcd}
			\hspace{0mm}
		\begin{tikzcd}[scalenodes=0.9, column sep=2em, row sep = .9em]
	    {(I A)  (-)} 
	    & 
	    {I  (A  -)} \\
	    {A (-)}
	    \arrow["\alpha_{I, A, -}", from=1-1, to=1-2]
	    \arrow[""{name=0, anchor=center, inner sep=0}, bend left, "\lambda", from=1-2, to=2-1]
	    \arrow[""{name=1, anchor=center, inner sep=0}, "{\lc{\lambda}}"', from=1-1, to=2-1]
	    \arrow["{\mathfrak{l}_{A, -}}"{}, 
			    					shift left=0,
			    					yshift = 1mm,
			    					shorten <=6pt, 
			    					shorten >=18pt, 
			    					Rightarrow, 
			    					from=1, 
			    					to=0]
			\end{tikzcd}
			\hspace{-1mm}
			\begin{tikzcd}[scalenodes=0.9, column sep=2em, row sep = .9em]
				(-B)I 
				\arrow{r}{\alpha_{-, B, \tensu}}
				\arrow[bend right, name=0]{dr}[swap]{\rho_{-B}}
				&
				(-)(BI)
				\arrow{d}[xshift=.5mm]{\rc{\rho}}
				\\
				\:
				\arrow["{\mathfrak{r}_{-, B}}",
							Rightarrow, 
							xshift=5mm,
							yshift=1mm,
							shorten <=16pt, 
							shorten >=16pt, 
							from=2-1, 
							to = 1-2]
				&
				AB
				\end{tikzcd}
				\vspace{-2mm}
				\end{center}

%% file: diag-bowtie-on-centre.tex
\begin{tikzcd}
	&& {} \\
	{} &&& {} \\
	&& {} \\
	{} &&& {} \\
	& {}
	\arrow["{(Xf)B}", from=2-1, to=1-3]
	\arrow[""{name=0, anchor=center, inner sep=0}, "{(XA')g}"{description}, from=1-3, to=3-3]
	\arrow["\alpha", from=1-3, to=2-4]
	\arrow["\alpha"{description}, from=3-3, to=4-4]
	\arrow[""{name=1, anchor=center, inner sep=0}, "{X(A'g)}", from=2-4, to=4-4]
	\arrow[""{name=2, anchor=center, inner sep=0}, "{X(fB')}"', from=5-2, to=4-4]
	\arrow["\alpha"', from=4-1, to=5-2]
	\arrow[""{name=3, anchor=center, inner sep=0}, "{(Xf)B'}"{description}, from=4-1, to=3-3]
	\arrow[""{name=4, anchor=center, inner sep=0}, "{(XA)g }"', from=2-1, to=4-1]
	\arrow["{\lc{X\ltie f}_g}"{description}, draw=none, from=4, to=0]
	\arrow["{\cellOf{\alpha}}"{description}, draw=none, from=0, to=1]
	\arrow["{\cellOf{\alpha}}"{description}, draw=none, from=3, to=2]
      \end{tikzcd}
    \quad  = \quad\begin{tikzcd}
	&& {} \\
	{} &&& {} \\
	& {} \\
	{} &&& {} \\
	& {}
	\arrow[""{name=0, anchor=center, inner sep=0}, "{(Xf)B}", from=2-1, to=1-3]
	\arrow["\alpha", from=1-3, to=2-4]
	\arrow[""{name=1, anchor=center, inner sep=0}, "{X(A'g)}", from=2-4, to=4-4]
	\arrow["{X(fB')}"', from=5-2, to=4-4]
	\arrow["\alpha"', from=4-1, to=5-2]
	\arrow[""{name=2, anchor=center, inner sep=0}, "{(XA)g }"', from=2-1, to=4-1]
	\arrow[""{name=3, anchor=center, inner sep=0}, "{X(f B)}"{description}, from=2-4, to=3-2]
	\arrow["\alpha"{description}, from=2-1, to=3-2]
	\arrow[""{name=4, anchor=center, inner sep=0}, "{X(Ag')}"{description}, from=3-2, to=5-2]
	\arrow["{X\ltie\lc{f}_g}"{description}, draw=none, from=4, to=1]
	\arrow["{\cellOf{\alpha}}"{description}, draw=none, from=2, to=4]
	\arrow["{\cellOf{\alpha}}"{description}, draw=none, from=0, to=3]
\end{tikzcd}

%% file: sec-freyd-bicategories.tex
We build up to our definition of Freyd bicategories in
stages. Although the
bicategories of values and computations have the same objects and
their structures are tightly connected, 
bicategories offer a range of levels of strictness, so we must make
careful choices. 

We begin with a useful technical notion for relating two pseudofunctors
which agree on objects:

\begin{definition}[\cite{Lack2008Icons}]
For pseudofunctors $F, G : \B \to \C$ which agree on
objects, an \emph{icon} $\theta : F \to G$ is an oplax natural transformation
whose 1-cell components are all identity. More explicitly, $\theta$
is a family of 2-cells $\theta_f : F(f) \to G(f)$ indexed by 1-cells of
$\B$, subject to naturality, identity and composition laws. 
\end{definition}

Using this, we define a notion of strict morphism between binoidal bicategories.

\vspace{1mm}
\noindent
\begin{minipage}{0.7\textwidth}
\begin{definition}
Let $\binoidal{\V}$ and $\binoidal{\B}$ be binoidal bicategories. A 
	\emph{0-strict binoidal pseudofunctor} 
is a pseudofunctor $\J : \V \to \B$
together  with families of invertible icons $\theta^A$ and
$\zeta^A$ (for $A \in \B$) as on the right; their existence implicitly requires that $\J(A \tens B) = \J A \tens \J B$. 
\end{definition}
\end{minipage}
\hfill
\begin{minipage}{0.3\textwidth}
	\vspace{-3mm}
\[
	\begin{tikzcd}[
		row sep = 2em, 
		column sep = 3em,
		execute at end picture={
						\foreach \nom in  {A,B,C, D}
			  				{\coordinate (\nom) at (\nom.center);}
						\fill[\extcolour,opacity=\opacity] 
			  				(A) -- (B) -- (C) -- (D) -- (A);
	}   			
	]
		\alias{A} \V
		\arrow[phantom]{dr}{\twocellIso\zeta^A}
		\arrow{r}{ (-) \rtie A}
		\arrow[swap]{d}{\J}
		&
		\V 
		\arrow{d}[description]{\J}
		&
		\alias{D} \V  
		\arrow[phantom]{dl}{\twocellIso\theta^A}
		\arrow{l}[swap]{A \ltie (-)}
		\arrow{d}{\J}
		\\
		\alias{B} \B 
		\arrow[swap]{r}{(-) \rtie \J A}
		&
		\B 
		&
		\alias{C} \B
		\arrow{l}{\J A \ltie (-)}
	\end{tikzcd}
      \]
\end{minipage}
\vspace{1mm}

It is crucial that we take preservation up to icons, and not up to identity. In the context of \Cref{res:premonoidal-structure-from-writer-monad}, for instance, we get a 0-strict binoidal pseudofunctor 
	$S \tens (-) :  \B \to \K$
with icons $\theta$ and $\zeta$ constructed using the pseudonaturality of the equivalences in (\ref{eq:binoidal-structure-for-State}).
However, these icons do strictly commute with the premonoidal structure of $\K$ by the coherence of symmetric monoidal bicategories~\cite{Gurski2013symmetricbicats}.
This suggests the following; for simplicity we focus on the case where $\J$ is identity-on-objects.  

\begin{definition}
\label{def:premonoidal-pseudofunctor}
Let $\premonoidal{\V}$ and $\premonoidal{\B}$ be premonoidal
bicategories with the same objects and unit $I$. An \emph{identity-on-objects, 
	0-strict premonoidal pseudofunctor}  
$\V \to \B$ is a 0-strict binoidal pseudofunctor $(\J, \theta, \zeta)$ such that $\J$ is
identity-on-objects and the following axioms hold:
\begin{enumerate}
\item 
	$\J$ strictly preserves the components of the structural transformations:
		for each $A, B, C \in \B$ we have 
			$\J \alpha_{A,B,C} =  \alpha_{A,B,C}, \J \lambda_A = \lambda_A$, 
		and $\J\rho_A = \rho_A$;
\item
	$\J$ preserves structural 2-cells up to the icons $\theta$ and $\zeta$, according to the axioms in 
		\Cref{sec:equations-for-freyd-bicats}.
\end{enumerate}
\end{definition}

A Freyd bicategory is an identity-on-objects 
	0-strict premonoidal pseudofunctor from a monoidal bicategory
        of values
        to a premonoidal bicategory of computations,
        together with a choice of centrality witnesses for every
        value. This choice must be functorial, coherent, and compatible
        with the interchange law whenever two values are being
        interchanged. We formalize this in terms of a strict
        factorization through the centre $\centreOf \B$, as is done for
        Freyd categories \cite{Levy2003}. (Unlike for Freyd categories,
        this factorization is additional structure and not a
        property of the premonoidal pseudofunctor.)

\begin{definition}
  \label{def:freyd-bicategory}
  A \emph{Freyd bicategory} $\freydCat$ consists of a monoidal
  bicategory $\monoidal\V$, a premonoidal bicategory, $\premonoidal{\B}$, an
  identity-on-objects, 0-strict premonoidal pseudofunctor $\J : \V \to \B$, and a binoidal pseudofunctor $\J_{\centre}$ factoring $\J$ through the
  centre of $\B$, as pictured below, 
    \[\begin{tikzcd}
      \V && \B \\
      & \centreOf\B
      \arrow["\mathrm{forget}"', from=2-2, to=1-3]
      \arrow["\J", from=1-1, to=1-3]
      \arrow["{\J_\centre}"', dashed, from=1-1, to=2-2]
    \end{tikzcd}\]
 such that the following axioms hold, where we write 
 	$(\J f, \lc{\J f}, \rc{\J f})$ 
 for
  	$\J_\centre(f)$:
 \begin{enumerate}
 \item The chosen centrality witnesses for the structural 1-cells agree with those in the premonoidal structure of $\B$.%

 \item 	
 	For each value $f : A \to A'$, the chosen $\lc{\J f}$ satisfies the following
        compatibility law for every $g \in \B(B, B')$ and $X\in
        \B$. (This complements  the constructions of
        \Cref{res:assoc-binoidal-bicat-has-binoidal-centre}.)
\[\begin{tikzcd}
	&& {} \\
	{} &&& {} \\
	&& {} \\
	{} &&& {} \\
	& {}
	\arrow["{(fB)X}", from=2-1, to=1-3]
	\arrow["\alpha", from=1-3, to=2-4]
	\arrow[""{name=0, anchor=center, inner sep=0}, "{A'(gX)}", from=2-4, to=4-4]
	\arrow["{f(B'X)}"', from=5-2, to=4-4]
	\arrow[""{name=1, anchor=center, inner sep=0}, "\alpha"', from=4-1, to=5-2]
	\arrow[""{name=2, anchor=center, inner sep=0}, "{(Ag)X }"', from=2-1, to=4-1]
	\arrow[""{name=3, anchor=center, inner sep=0}, "{(A'g)X}"{description}, from=1-3, to=3-3]
	\arrow["{(fB')X}"{description}, from=4-1, to=3-3]
	\arrow[""{name=4, anchor=center, inner sep=0}, "\alpha"{description}, from=3-3, to=4-4]
	\arrow["{\cellOf{\alpha}}"{description}, draw=none, from=1, to=4]
	\arrow["{\cellOf{\alpha}}"{description}, draw=none, from=3, to=0]
	\arrow["{\lc{\J f}_g \rtie X}"{description}, draw=none, from=2, to=3]
      \end{tikzcd}
      \quad = \quad
      \begin{tikzcd}
	&& {} \\
	{} &&& {} \\
	& {} \\
	{} &&& {} \\
	& {}
	\arrow[""{name=0, anchor=center, inner sep=0}, "{(fB)X}", from=2-1, to=1-3]
	\arrow["\alpha", from=1-3, to=2-4]
	\arrow[""{name=1, anchor=center, inner sep=0}, "{A'(gX)}", from=2-4, to=4-4]
	\arrow["{f(B'X)}"', from=5-2, to=4-4]
	\arrow["\alpha"', from=4-1, to=5-2]
	\arrow[""{name=2, anchor=center, inner sep=0}, "{(Ag)X }"', from=2-1, to=4-1]
	\arrow[""{name=3, anchor=center, inner sep=0}, "{f(BX)}"{description}, from=2-4, to=3-2]
	\arrow["\alpha"{description}, from=2-1, to=3-2]
	\arrow[""{name=4, anchor=center, inner sep=0}, "{A(gX)}"{description}, from=3-2, to=5-2]
	\arrow["{\cellOf{\alpha}}"{description}, draw=none, from=2, to=4]
	\arrow["{\lc{\J f}_{g\rtie X}}"{description}, draw=none, from=4, to=1]
	\arrow["{\cellOf{\alpha}}"{description}, draw=none, from=0, to=3]
\end{tikzcd}\]
 \item \label{c:l-equals-r}
   For values 
   	$x : X \to X'$ and $y : Y \to Y'$,  
   $\lc{\J x}_{\J y}$ and $\rc{\J y}_{J x}$ are determined by the interchange law in $\V$~(\ref{eq:centrality-from-interchange}):
   \vspace{-2mm}
   \[
   \vspace{-2mm}
\begin{tikzcd}[
		execute at end picture={
						\foreach \nom in  {A,B,C}
			  				{\coordinate (\nom) at (\nom.center);}
						\fill[\extcolour,opacity=\opacity] 
			  				(A) to[curve={height=12pt}] (B) to[curve={height=12pt}] (A);
						\fill[\extcolour,opacity=\opacity] 
			  				(B) to[curve={height=12pt}] (C) to[curve={height=12pt}] (B);
	}   			
	]
	\alias{A} XY & \alias{B} {XY'} & \alias{C} {X'Y'} \\
	& {X'Y}
	\arrow["{\J(x) \rtie Y}"', curve={height=6pt}, from=1-1, to=2-2]
	\arrow["{\J(Xy)}", curve={height=-10pt}, from=1-1, to=1-2]
	\arrow["{X' \ltie \J y}"', curve={height=6pt}, from=2-2, to=1-3]
	\arrow["{\J(xY')}", curve={height=-10pt}, from=1-2, to=1-3]
	\arrow[curve={height=10pt}, from=1-2, to=1-3]
	\arrow[curve={height=10pt}, from=1-1, to=1-2]
	\arrow["\theta"{description}, draw=none, from=1-1, to=1-2]
	\arrow["\zeta"{description}, draw=none, from=1-2, to=1-3]
	\arrow["{\rc{\J y}_{\J x}}"{description, yshift = -1mm}, draw=none, from=1-2, to=2-2]
	\arrow[
			swap,
			rounded corners,
			to path=
			{ -- ([yshift=.8cm]\tikztostart.north)
			-| ([yshift=.8cm]\tikztotarget.south)
			-- ([yshift=.6cm]\tikztotarget.south)
			}, 
			from=1-1, to=1-3
		]		
	\arrow[
			"\iso"{description, yshift=6mm},
			from=1-1, to=1-3,
			draw = none,
		]		
	\arrow[
			"\J(xy)"{description, yshift=12mm},
			from=1-1, to=1-3,
			draw = none,
		]				
\end{tikzcd}
\hspace{2mm}
=
\hspace{2mm}
\begin{tikzcd}[
		execute at end picture={
						\foreach \nom in  {A,B,C}
			  				{\coordinate (\nom) at (\nom.center);}
						\fill[\extcolour,opacity=\opacity] 
			  				(A) to[curve={height=12pt}] (B) to[curve={height=12pt}] (A);
						\fill[\extcolour,opacity=\opacity] 
			  				(B) to[curve={height=12pt}] (C) to[curve={height=12pt}] (B);
	}   			
	]
	\alias{A} XY  & \alias{B} {X'Y} & \alias{C} {X'Y'}
	\arrow[""{name=0, anchor=center, inner sep=0}, "{\J(xY)}", curve={height=-10pt}, from=1-1, to=1-2]
	\arrow[""{name=1, anchor=center, inner sep=0}, "{\J(X'y)}", curve={height=-10pt}, from=1-2, to=1-3]
	\arrow[""{name=2, anchor=center, inner sep=0}, "{X' \ltie \J y}"', curve={height=10pt}, from=1-2, to=1-3]
	\arrow[""{name=3, anchor=center, inner sep=0}, "{\J(x) \rtie Y}"', curve={height=10pt}, from=1-1, to=1-2]
	\arrow["\theta"{description}, draw=none, from=1, to=2]
	\arrow["\zeta"{description}, draw=none, from=0, to=3]
	\arrow[
			swap,
			rounded corners,
			to path=
			{ -- ([yshift=.8cm]\tikztostart.north)
			-| ([yshift=.8cm]\tikztotarget.south)
			-- ([yshift=.6cm]\tikztotarget.south)
			}, 
			from=1-1, to=1-3
		]			
	\arrow[
			"\iso"{description, yshift=6mm},
			from=1-1, to=1-3,
			draw = none,
		]		
	\arrow[
			"\J(xy)"{description, yshift=12mm},
			from=1-1, to=1-3,
			draw = none,
		]					
\end{tikzcd}
\hspace{2mm}
=
\hspace{2mm}
\begin{tikzcd}[
		execute at end picture={
						\foreach \nom in  {A,B,C}
			  				{\coordinate (\nom) at (\nom.center);}
						\fill[\extcolour,opacity=\opacity] 
			  				(A) to[curve={height=12pt}] (B) to[curve={height=12pt}] (A);
						\fill[\extcolour,opacity=\opacity] 
			  				(B) to[curve={height=12pt}] (C) to[curve={height=12pt}] (B);
	}   			
	]
	\alias{A} XY & \alias{B} {XY'} & \alias{C} {X'Y'} \\
	& {X'Y}
	\arrow["{\J(x) \rtie Y}"', curve={height=6pt}, from=1-1, to=2-2]
	\arrow["{\J(Xy)}", curve={height=-10pt}, from=1-1, to=1-2]
	\arrow["{X' \ltie \J y}"', curve={height=6pt}, from=2-2, to=1-3]
	\arrow["{\J(xY')}", curve={height=-10pt}, from=1-2, to=1-3]
	\arrow[curve={height=10pt}, from=1-2, to=1-3]
	\arrow[curve={height=10pt}, from=1-1, to=1-2]
	\arrow["\theta"{description}, draw=none, from=1-1, to=1-2]
	\arrow["\zeta"{description}, draw=none, from=1-2, to=1-3]
	\arrow["{\lc{\J x}_{\J y}}"{description, yshift=-1mm}, draw=none, from=1-2, to=2-2]
	\arrow[
			swap,
			rounded corners,
			to path=
			{ -- ([yshift=.8cm]\tikztostart.north)
			-| ([yshift=.8cm]\tikztotarget.south)
			-- ([yshift=.6cm]\tikztotarget.south)
			}, 
			from=1-1, to=1-3
		]		
	\arrow[
			"\iso"{description, yshift=6mm},
			from=1-1, to=1-3,
			draw = none,
		]		
	\arrow[
			"\J(xy)"{description, yshift=12mm},
			from=1-1, to=1-3,
			draw = none,
		]					
\end{tikzcd}
\]
 \end{enumerate}        
\end{definition}

\subsection{Examples of Freyd bicategories}
\label{sec:examples-of-Freyd-bicategories}

Two of the examples of premonoidal bicategories from \Cref{sec:first-examples-of-premonoidal-bicats} naturally yield Freyd bicategories. 
First, in the context of \Cref{res:premonoidal-structure-from-writer-monad}, we have a pseudofunctor 
	$S \tens (-) : \B \to \K$ 
and icons $\theta$ and $\zeta$ constructed using the equivalences defining the binoidal structure 	
	(recall~(\ref{eq:binoidal-structure-for-State})).
Moreover, coherence for symmetric monoidal bicategories~\cite{Gurski2013symmetricbicats} gives a unique choice of 2-cell for each $\lc{S \tens f}_g$ and $\rc{S \tens f}_g$, so $S \tens (-)$ factors through the centre, yielding the following.

\begin{restatable}
	{lemma}
	{FreydFromState}
\label{res:Freyd-bicategory-from-state-monad}
Let 
	$\monoidal{\B}$
be a symmetric monoidal bicategory and let $\K$ be the premonoidal bicategory defined in 
	\Cref{res:premonoidal-structure-from-writer-monad}.
Then the pseudofunctor 
	$S \tens (-)$
defines a Freyd bicategory
	$\B \to \K$.
\end{restatable}

Similarly, for a bistrong graded monad $T$, we can think of morphisms
in the base monoidal category $\catC$ as parameterized maps with
trivial parameter space, to construct a Freyd bicategory. The
identity-on-objects pseudofunctor has action on morphisms determined by
$\J(f) := \para{f} = \eta \circ f$.
The structural icons $\theta$ and $\zeta$ are the identity, and $\J$ factors strictly through the centre because every $ \para{f}$ has a canonical choice of centrality data.

\begin{restatable}{proposition}{FreydOnGradedKleisli}
Let $(T, \eta, \mu)$ be a bistrong graded monad on $\monoidal\catC$ 
	with grades $(\catE, \bullet, \tensu)$.
Then, writing $\d\catC$ for the monoidal category $\catC$ viewed as a locally-discrete monoidal 2-category, there exists a canonical choice of pseudofunctor $\J$ making 
	$\J : \d\catC \to \Kl_T$
a Freyd bicategory. 
\end{restatable}

Finally, recall the unnatural
transformations discussed in \Cref{sec:first-examples-of-premonoidal-bicats}: although one could expect the inclusion 
	$\iota : [\B, \B] \hookrightarrow [\B, \B]_u$
to be a Freyd bicategory, this is not true even in the categorical
setting: it is not the case that every natural transformation is
central, so $\iota$ does not factor through the centre.

%% file: actions.tex
Freyd categories may equivalently be defined as certain actions of monoidal
categories (\eg~\cite{LevyBook}). In this section we show that this is also possible
in the two-dimensional setting. 

We first define actions of monoidal bicategories.
 As observed in~\cite{JanelidzeKelly2001}, a left action on a category is equivalently a bicategory with two objects and certain hom-categories taken to be trivial. We therefore define a left action on a bicategory so it is equivalently a \emph{tricategory} (see~\cite{Gordon1995}) with two objects and certain hom-bicategories taken to be trivial. 
 It follows from the coherence of
 tricategories~(\cite{Gordon1995,Gurski2013}) that every diagram of
 2-cells constructed using the structural data of an action must
 commute. 

\begin{definition}
\label{def:action-of-monoidal-bicategory}
  A \emph{left action} of a monoidal bicategory $(\V, \otimes, I)$ on a
  bicategory $\B$ consists of a pseudo-functor 
  	$\act : \V \times \B \to \B$, 
  together with the following data:
  \begin{itemize}%
  \item 
  	Pseudonatural adjoint equivalences %
     $\actlambda_A : I \act A \to A$ and 
     $\actalpha_{X,Y,C} : {(X \otimes Y) \act C} \to X \act (Y \act C)$; 
  \item 
  	Invertible modifications as shown below,
  	satisfying the same coherence axioms as $\pentagonator, \montrianglem$, and $\montrianglel$ in a monoidal bicategory (\eg~\cite{Stay2016}):
\end{itemize}
	\input{diag-equations-for-defining-actions}
\end{definition}

A \emph{right} action $\ract : \B \times \V \to \B$
can be defined analogously, with a right unitor 
	$\actrho_A : A \ract I \to A$, an
associator 
	$\actalpha_{A, X, Y} : (A \ract X) \ract Y \to A \ract (X \tens Y)$, 
and 2-dimensional structural data. 

\newcommand{\Vact}{\V\text{-}\mathbf{act}_{0s}}

Every monoidal bicategory $\V$ has canonical left and right actions on
itself given by the monoidal data. As we will see, a Freyd bicategory $\J : \V \to \B$
corresponds to a pair of actions $\act : \V \times \B \to \B$ and 
	$\ract : \B \times \V \to \B$ 
that extend the canonical actions: this mirrors the categorical situation. 
To that end, we consider a category $\Vact$ of actions of $\V$ and
identity-on-objects pseudofunctors that preserve the action strictly on objects, but weakly on morphisms. (This is a very special
case of a more canonical notion of map between actions.)

\vspace{\baselineskip}
\noindent
\begin{minipage}{0.75\textwidth}
\begin{definition}
Let $\V$ be a monoidal bicategory and let $(\B, \lact)$ and 
	$(\B', \lactprime)$ be left actions of $\V$. A \emph{0-strict morphism of actions} from 
	$(\B, \lact)$ to $(\B', \lactprime)$ is an
identity-on-objects functor $\J : \B \to \B'$ 
satisfying 
		$\widetilde{\lambda^{\lact}_A} = \J(\widetilde{\lambda^{\lactprime}_A})$
and 
		$\widetilde{\alpha^{\lact}}_{\!\!A, B, C} = \J(\widetilde{\alpha^{\lactprime}}_{\!\!A, B, C})$
for every $A, B, C \in \B$, 
equipped with an icon as on the right, 
which relates the structural data for the actions according to the
axioms below: 
\noindent
\end{definition}
\end{minipage}
\hfill
\begin{minipage}{0.2\textwidth}
	\centering
   \begin{tikzcd}[
		execute at end picture={
							\foreach \nom in  {A,B,C, D}
				  				{\coordinate (\nom) at (\nom.center);}
							\fill[\extcolour,opacity=\opacity] 
				  				(A) -- (C) -- (D) -- (B) -- (A);
	}   
   ]
     \alias{A}
     \V \times \B
     \arrow{r}{\act} 
     \arrow[name=0]{d}[swap]{\V \times \J}
    & 
    \alias{B}
    \V
    \arrow[name=1]{d}{\J} \\
    \alias{C}
     \V \times \B'
     \arrow[swap]{r}{\lactprime}
 	 & 
 	 \alias{D}
 	 \B 
    \arrow[Rightarrow, shorten = 8, "\theta", from=C, to=B, xshift=0mm,yshift=-1mm] 
  \end{tikzcd}
\end{minipage}
\vspace{-6mm}
\input{diag-extension-of-an-action-equations}
\vspace{0mm}

A key example is the following:
\begin{definition}
\label{def:extending-canonical-action}
For a monoidal bicategory $(\V, {}{\otimes}{}, I)$, 
a \emph{left extension of the canonical action of $\V$ on itself}
is a $\V$-action $(\B, \act)$, together with a 0-strict morphism $(\J,
\theta) : (\V, \tens) \to (\B, \act)$ such that $\theta$ is
invertible. (We say this is an extension \emph{along $\J$}.)
\end{definition}

\begin{figure}[!t]
\centering

\input{diag-Freyd-bicat-equations}

\caption{Compatibility laws for \Cref{def:premonoidal-pseudofunctor}}
\label{sec:equations-for-freyd-bicats}
\end{figure}

We define a \emph{right extension} analogously; this involves a right action 
	$\ract : \B \times \V \to \B$ 
and an invertible icon with components
	$\zeta_{f, g} : f \ract \J g \To \J(f \tens g)$. 
The rest of this section is devoted to showing 
Freyd bicategories may be equivalently presented as pairs of extensions, which we call \emph{Freyd actions}. 

\begin{definition}
\label{def:freyd-action}
A \emph{Freyd action} consists of an identity-on-objects pseudofunctor 
	$\J : \V \to \B$
from a monoidal bicategory 
	$\monoidal\V$
to a bicategory $\B$, together with: 
\begin{enumerate}
\item 
	A left extension $(\lact, \theta)$ and right extension $(\ract, \zeta)$ along $\J$ of the canonical actions of $\V$ on itself;
\item 
	A pseudonatural adjoint equivalence 
		$\kappa$
	with 1-cell components
	$\kappa_{X, B, Z} = \J(\alpha_{X, B, Z}) : (X \lact B) \ract Z \to X \lact (B \ract Z)$,
	subject to the equation below 
	and additional axioms given in \Cref{app:missingaxioms}:
	\[
		\begin{tikzcd}[
			column sep = 3em, 
			row sep = 2em,
			execute at end picture={
								\foreach \nom in  {A,B,C, D,E, X, Y}
					  				{\coordinate (\nom) at (\nom.center);}
								\fill[\extcolour,opacity=\opacity] 
					  				(A) to[curve={height=35pt}] (B) to[curve={height=35pt}] (A);
								\fill[\kcolour,opacity=\opacity] 
					  				(X) to (Y)  to (B) to[curve={height=-35pt}] (A) to (X);
			}
		]
		\alias{X} {(XY)Z} & \alias{A} {X(YZ)} \\
		\alias{Y} {(X'Y')Z'} & \alias{B} {X'(Y'Z')}
		\arrow[""{name=0, anchor=center, inner sep=0}, "{(f \lact Jg) \ract h}"', from=1-1, to=2-1]
		\arrow["\kappa", from=1-1, to=1-2]
		\arrow[""{name=1, anchor=center, inner sep=0}, curve={height=28pt}, from=1-2, to=2-2]
		\arrow["\kappa"', from=2-1, to=2-2]
		\arrow[""{name=2, anchor=center, inner sep=0}, from=1-2, to=2-2]
		\arrow[""{name=3, anchor=center, inner sep=0}, "{\J(f \tens (g \tens h))}"{}, curve={height=-28pt}, from=1-2, to=2-2]
		\arrow["{\cellOf{\kappa}_{f, Jg, h}}"{description}, draw=none, from=0, to=1]
		\arrow["{X' \lact \zeta}"{description}, shorten <=2pt, shorten >=2pt, Rightarrow, from=1, to=2]
		\arrow["\theta"{description}, draw=none, from=2, to=3]
		\end{tikzcd}	
		\hspace{2mm}
		=
		\hspace{2mm}
		\begin{tikzcd}[
			column sep = 3em, 
			row sep = 2em,
			execute at end picture={
					\foreach \nom in  {A,B,C, D,E}
		  				{\coordinate (\nom) at (\nom.center);}
					\fill[\extcolour,opacity=\opacity] 
		  				(A) to[curve={height=35pt}] (B) to[curve={height=35pt}] (A);
		}		
		]
			\alias{A} {(XY)Z} & {X(YZ)} \\
			\alias{B} {(X'Y')Z'} & {X'(Y'Z')}
			\arrow[""{name=0, anchor=center, inner sep=0}, "{(f \lact Jg) \ract h}"', curve={height=28pt}, from=1-1, to=2-1]
			\arrow["{\J(\alpha)}", from=1-1, to=1-2]
			\arrow["{\J(\alpha)}"', from=2-1, to=2-2]
			\arrow[""{name=1, anchor=center, inner sep=0}, "{\J(f \tens (g \tens h))}", from=1-2, to=2-2]
			\arrow[""{name=2, anchor=center, inner sep=0}, from=1-1, to=2-1]
			\arrow[""{name=3, anchor=center, inner sep=0}, curve={height=-28pt}, from=1-1, to=2-1]
			\arrow["{\J(\cellOf{\alpha}_{f,g,h})}"{description}, draw=none, shorten <=6pt, shorten >=6pt, Rightarrow, from=3, to=1]
			\arrow["{\theta \ract Z}"{description}, draw=none, from=0, to=2]
			\arrow["\zeta"{description}, draw=none, shorten <=4pt, shorten >=4pt, Rightarrow, from=2, to=3]
		\end{tikzcd}
	\]
\end{enumerate}
\end{definition}

\label{sec:correspondence-theorem}

We construct an equivalence of categories between Freyd actions and
Freyd bicategories, over a fixed identity-on-objects pseudofunctor $\J
: \V \to\B$. (The corresponding categorical result is a
bijection, but we must work modulo the structural
isomorphisms, and hence lose the strictness.)

On one side, the category $\FreydAct{}(\J)$ has
objects Freyd actions $(\lact, \theta, \ract, \zeta, \kappa)$ with
underlying pseudofunctor $\J$. 
Morphisms 
	$((\lact, \theta), (\ract, \zeta), \kappa) \to ((\lact', \theta'), (\ract', \zeta), \kappa')$
are pairs of icons $\lefttrans : \lact \To \lact'$ and
 		$\righttrans : \ract \To \ract'$
fitting in the diagram in $\Vact$ as on the left below, such that
$\kappa$ is preserved as 
on the right:
\[
\hspace{-2mm}
\begin{tikzcd}[column sep = 1.5em]
    {(\B,\lact)} & {(\V, \tens)} & (\B, \ract) \\
    {(\B, \lact')} & \: & (\B', \ract')
    \arrow["{(\id_{\B}, \lefttrans)}"', from=1-1, to=2-1]
    \arrow["{(J, \theta')}", from=1-2, to=2-1]
    \arrow["{(J, \theta)}"', from=1-2, to=1-1]
	\arrow["{(\id_{\B}, \righttrans)}"{}, from=1-3, to=2-3]
	\arrow["{(J, \zeta')}", from=1-2, to=1-3]
	\arrow["{(J, \zeta)}"', from=1-2, to=2-3]    
  \end{tikzcd}
\hspace{4mm}
  \begin{tikzcd}[
    column sep=1em, scalenodes=1,
    execute at end picture={
      \foreach \nom in  {A,B,C, D, Y, X}
      {\coordinate (\nom) at (\nom.center);}
      \fill[\extcolour,opacity=\opacity] 
      (A) to[curve={height=-30pt}] (B) to[curve={height=-30pt}] (A);
      \fill[\kcolour,opacity=\opacity] 
      (A) to[curve={height=-30pt}] (B) -- (Y) -- (X) -- (A); 
    }
    ]
    \alias{A} {(A B) C} && \alias{X} A(B C) \\
    \alias{B} {(A' B') C'} && \alias{Y} {A' (B' C')}
    \arrow["{\kappaprime}", from=1-1, to=1-3]
    \arrow["{\kappaprime}"', from=2-1, to=2-3]
    \arrow[""{name=0, anchor=center, inner sep=0}, "{f \lact' (b \ract' h)}", 
    					from=1-3, to=2-3]
    \arrow[""{name=1, anchor=center, inner sep=0}, 
    curve={height=24pt}, from=1-1, to=2-1, ""{swap}] %
    \arrow[""{name=2, anchor=center, inner sep=0}, %
    from=1-1, to=2-1]
    \arrow[""{name=3, anchor=center, inner sep=0}, curve={height=0pt}, from=1-1, to=2-1]
    \arrow[""{name=4, anchor=center, inner sep=0}, 
    curve={height=-24pt}, from=1-1, to=2-1] 			
    \arrow["\righttrans"{description}, draw=none, from=4, to=3]			
    \arrow["\lefttrans \ract h"{yshift=-2mm, xshift=.2mm}, draw=none, from=1, to=3]	
    \arrow["\oncell{\cellOf{(\kappaprime)}_{f, b, h}}"{description, xshift=4mm}, 
    		draw=none, from=2, to=0] 		  
   \end{tikzcd}
\hspace{0mm}
  =
\hspace{0mm}
  \begin{tikzcd}[
    column sep=1em, scalenodes=1,
    execute at end picture={
      \foreach \nom in  {A,B,C, D, X, Y}
      {\coordinate (\nom) at (\nom.center);}
      \fill[\extcolour,opacity=\opacity] 
      (A) to[curve={height=28pt}] (B) to[curve={height=28pt}]  (A);
      \fill[\kcolour,opacity=\opacity] 
      (A) to[curve={height=28pt}]  (B) -- (Y) -- (X);
    }
    ]
    \alias{X} {(AB)C} && \alias{A} {A(BC)} \\
    \alias{Y} {(A'B')C'} && \alias{B} {A'(B'C')}
    \arrow[""{name=0, anchor=center, inner sep=0}, "\kappaone", from=1-1, to=1-3]
    \arrow["\kappa"', from=2-1, to=2-3]
    \arrow[""{name=1, anchor=center, inner sep=0}, 
    curve={height=22pt}, from=1-3, to=2-3] 		%
    \arrow[""{name=2, anchor=center, inner sep=0}, %
    from=1-1, to=2-1, "(f \lact b) \ract h"{swap}]
    \arrow[""{name=3, anchor=center, inner sep=0}, curve={height=0pt}, from=1-3, to=2-3]
    \arrow[""{name=4, anchor=center, inner sep=0}, "", 
    curve={height=-22pt}, from=1-3, to=2-3] %
    \arrow["\oncell{{\cellOf{\kappaone}_{f, b, h}}}"{description}, draw=none, from=2, to=1]
    \arrow["\oncell{{f \lact \righttrans}}"{yshift=-2mm,xshift=.17mm}, draw=none, from=1, to=3]
    \arrow["\oncell{{\lefttrans}}"{description}, draw=none, from=3, to=4]
  \end{tikzcd}
\]

On the other side, the category $\FreydBicat{\V, \J, \B}(\J)$ has
objects Freyd bicategories whose underlying pseudofunctor is $\J$;
these are determined by a premonoidal structure on $\B$ and families of
icons $\theta = \{ \theta^A \st A \in \B \}$ and $\{ \zeta^A \st A \in \B \}$ 
making the pseudofunctor $\J$ premonoidal. 
Morphisms 
	$(\ltie, \rtie, \theta, \zeta) \to (\ltie', \rtie',\theta', \zeta')$ 
are families of icons $\lefttrans^A  : (A \ltie -) \To (A \ltie' -)$
	and  $\righttrans^A  : (- \rtie A) \To (- \rtie' A)$
making the identity pseudofunctor $\B \to \B$ premonoidal and such
that 
	$(J, \theta', \zeta') \circ (\id_{\B}, \lefttrans, \righttrans) =
			(J, \theta, \zeta)$ 
as premonoidal pseudofunctors. 

Our correspondence theorem is then as follows. %

\begin{restatable}{theorem}{CorrespondenceTheorem}
\label{res:correspondence-theorem}
For any monoidal bicategory $\monoidal{\V}$, bicategory $\B$, and
identity-on-objects pseudofunctor $\J : \V \to \B$, the categories 
	$\FreydAct{\V, \J, \B}(\J)$
and
	$\FreydBicat{\V, \J, \B}(\J)$
are equivalent. 
\end{restatable}

%% file: diag-equations-for-defining-actions.tex
\begin{tikzcd}
	[column sep=1em, 
		scalenodes=0.9,
		execute at end picture={
					\foreach \nom in  {A,B,C, D,E}
		  				{\coordinate (\nom) at (\nom.center);}
					\fill[\actpcolour,opacity=\opacity] 
		  				(A) -- (B) -- (C) -- (D) -- (E);
		}
		]
	\alias{A}
	{\left((X  Y)  Z\right)}  \act D
	\arrow{r}{\actalpha}
	\arrow[swap]{d}{\alpha \act D}
	\arrow[	"{\mathfrak{\actpentagonator}\:}"{left},
					from=2-2, 
					to=1-2,  
					Rightarrow, 
					shorten = 2pt
					]
	&
	\alias{B}
	(XY) \act (Z \act D)
	\arrow{r}{\actalpha}
	&
	\alias{C}
	X \act {\left(Y \act (Z \act D) \right)}
	\\
	\alias{E}
	{\left( X (YZ) \right)} \act D
	\arrow[swap]{rr}{\actalpha}
	&
	\:
	&
	\alias{D}
	X \act {\left((YZ) \act D\right)}
	\arrow[swap]{u}{X \act \actalpha}
\end{tikzcd}
\hspace{0mm}
\begin{tikzcd}[
	scalenodes=0.9,
	column sep = 0em,
	execute at end picture={
								\foreach \nom in  {A,B,C}
					  				{\coordinate (\nom) at (\nom.center);}
								\fill[\actmcolour,opacity=\opacity] 
					  				(A) -- (B) -- (C);
	  }
	]
\: &
\alias{C}
{X \lact C}
&
\: 
\\
\alias{A}
{(X  I) \act C} 
\arrow{ur}{\rho \act C}
\arrow[swap]{rr}{\actalpha}
&
\:
& 
\alias{B}
{X \act (I \act C)} 
\arrow[name=0]{ul}[swap]{X \act \actlambda}
\arrow["{\acttrianglem\:}"{}, 
 				shift right=1, 
 				shorten <=6pt, 
 				shorten >=8pt, 
 				Rightarrow, 
 				from=2-2, 
 				to=1-2]
\end{tikzcd}
\hspace{2mm}
\begin{tikzcd}[
	scalenodes=0.9,
	column sep=1em,
	execute at end picture={
			\foreach \nom in  {A,B,C}
  				{\coordinate (\nom) at (\nom.center);}
			\fill[\actlcolour,opacity=\opacity] 
  				(A) -- (B) to[bend left = 36] (C);
		}
	]
\alias{A}
{(I Y) \act  C} 
& 
\alias{B}
{I  \act (Y \act C)} \\
\alias{C}
{Y \act C}
\arrow["\actalpha", from=1-1, to=1-2]
\arrow[""{name=0, anchor=center, inner sep=0}, bend left = 25, "\actlambda", from=1-2, to=2-1]
\arrow[""{name=1, anchor=center, inner sep=0}, "{\lambda  \act C}"', from=1-1, to=2-1]
\arrow["{\acttrianglel}"{}, 
	    					shift left=1,
	    					yshift = 1mm,
	    					shorten <=6pt, 
	    					shorten >=6pt, 
	    					Rightarrow, 
	    					from=1, 
	    					to=0]
\end{tikzcd}		

%% file: diag-extension-of-an-action-equations.tex
\[
	\hspace{-8mm}
	\begin{minipage}{0.6\textwidth}
	\begin{align*}
		\begin{tikzcd}[
				ampersand replacement = \&,
				column sep=2.2em,
				row sep = 1.2em,				
				execute at end picture={
									\foreach \nom in  {A,B}
						  				{\coordinate (\nom) at (\nom.center);}
									\fill[\extcolour,opacity=\opacity] 
						  				(A) to[bend left=77] (B) to[bend left=75] (A);
									\fill[\Klambdacolour,opacity=\opacity] 
						  				(A) to[bend left=80] (B) -- (Y) -- (X);
			}
				] 
			\alias{A} {IB} \&\& \alias{X} B \\
			\alias{B} {I B'} \&\& \alias{Y} {B'}
			\arrow["{\J\lactlambda}", from=1-1, to=1-3]
			\arrow["{\J\lactlambda}"', from=2-1, to=2-3]
			\arrow[""{name=0, anchor=center, inner sep=0}, "{\J b}", from=1-3, to=2-3]
		        \arrow[""{name=1, anchor=center, inner sep=0}, bend left = 40, from=1-1, to=2-1]
			\arrow[""{name=2, anchor=center, inner sep=0}, "{I \lactprime \J b}"{swap}', 
								bend right = 40, from=1-1, to=2-1]
			\arrow["\oncell{{\theta}}"{description}, draw=none, from=2, to=1]
			\arrow["\oncell{\J\cellOf{(\lactlambda)}_{b}}"{description}, draw=none, from=1, to=0]
		      \end{tikzcd}
		      \hspace{2mm}
		      &=
		      \hspace{2mm}
		      \begin{tikzcd}[
		      ampersand replacement = \&,
		      	column sep=2em,
		      			execute at end picture={
		      								\foreach \nom in  {A,B,C, D}
		      					  				{\coordinate (\nom) at (\nom.center);}
		      								\fill[\actlambdacolour,opacity=\opacity] 
		      					  				(A) -- (B) -- (D) -- (C);
		      		}
		      ]
			\alias{A} {I B} \&\& \alias{B} B \\
			\alias{C} {I B'} \&\& \alias{D} {B'}
			\arrow["{\lactlambdaprime}", from=1-1, to=1-3]
			\arrow["{\lactlambdaprime}"', from=2-1, to=2-3]
			\arrow[""{name=0, anchor=center, inner sep=0}, "{\J b}", from=1-3, to=2-3]
			\arrow[""{name=1, anchor=center, inner sep=0}, "{I \lactprime \J b}"', curve={height=0pt}, from=1-1, to=2-1]
			\arrow["\oncell{\cellOf{(\toact{\lambda^{\lactprime}})}_{\J b}}"{description}, 
							draw=none, from=1, to=0]
		\end{tikzcd}
	\\[1mm]
\begin{tikzcd}[
		ampersand replacement = \&,
		row sep = 1.2em,		
		column sep=1.5em, scalenodes=1,
				execute at end picture={
									\foreach \nom in  {A,B,C, D}
						  				{\coordinate (\nom) at (\nom.center);}
									\fill[\extcolour,opacity=\opacity] 
						  				(A) to[curve={height=-20pt}] (B) to[curve={height=-20pt}] (A);
									\fill[\Klambdacolour,opacity=\opacity] 
						  				(A) to[curve={height=-20pt}] (B) -- (Y) -- (X);
			}
		]
	\alias{A} {(A B) C} \&\& \alias{X} A(B C) \\
	\alias{B} {(A' B') C'} \&\& \alias{Y} {A' (B' C')}
	\arrow["{\J\lactalpha}", from=1-1, to=1-3]
	\arrow["{\J\lactalpha}"', from=2-1, to=2-3]
	\arrow[""{name=0, anchor=center, inner sep=0}, "\J(f \lactprime (g \lactprime c))", 
							from=1-3, to=2-3]
	\arrow["\J{\cellOf{(\lactalpha)}_{f, g, \J c}}"{description, xshift=2mm}, draw=none, from=2, to=0]
	\arrow[""{name=1, anchor=center, inner sep=0}, 
					curve={height=16pt}, from=1-1, to=2-1, "(f \tens g) \lactprime \J c"{swap}] 
	\arrow[""{name=4, anchor=center, inner sep=0}, 
						curve={height=-16pt}, from=1-1, to=2-1] 			
	\arrow["\theta"{description}, draw=none, from=4, to=1]								
	\end{tikzcd}
      \hspace{1mm}
      &=
      \hspace{0mm}
	\begin{tikzcd}[
		ampersand replacement = \&,
		row sep = 1.2em,		
		column sep=1.5em, scalenodes=1,
		execute at end picture={
								\foreach \nom in  {A,B,C, D}
					  				{\coordinate (\nom) at (\nom.center);}
								\fill[\extcolour,opacity=\opacity] 
					  				(A) to[curve={height=32pt}] (B) to[curve={height=30pt}]  (A);
								\fill[\Klambdacolour,opacity=\opacity] 
					  				(A) to[curve={height=32pt}]  (B) -- (Y) -- (X);
		}
		]
	\alias{X} {(AB)C} \&\& \alias{A} {A(BC)} \\
	\alias{Y} {(A'B')C'} \&\& \alias{B} {A'(B'C')}
	\arrow[""{name=0, anchor=center, inner sep=0}, "\lactalphaprime", from=1-1, to=1-3]
	\arrow["\lactalphaprime"', from=2-1, to=2-3]
	\arrow[""{name=1, anchor=center, inner sep=0}, 
					curve={height=26pt}, from=1-3, to=2-3] 		%
	\arrow[""{name=2, anchor=center, inner sep=0}, %
        from=1-1, to=2-1]
	\arrow[""{name=4, anchor=center, inner sep=0}, "", 
						curve={height=-24pt}, from=1-3, to=2-3] %
	\arrow["\oncell{{\cellOf{(\lactalphaprime)}_{f, g, \J c}}}"{description}, draw=none, from=2, to=1]
	\arrow[""{name=3, anchor=center, inner sep=0}, curve={height=0pt}, from=1-3, to=2-3]	
	\arrow["\oncell{{\theta}}"{description, xshift=.17mm}, draw=none, from=4, to=3]
	\arrow["\oncell{f \lactprime \theta}"{yshift=-2.5mm}, draw=none, from=1, to=3]	
	\end{tikzcd}
	\\[1mm]
	\begin{tikzcd}[
			ampersand replacement = \&,
			row sep = 1.2em,			
	   		column sep=1em, scalenodes=1,
			execute at end picture={
								\foreach \nom in  {A,B,C, D}
					  				{\coordinate (\nom) at (\nom.center);}
								\fill[\extcolour,opacity=\opacity] 
					  				(A) to[curve={height=24pt}] (B) -- (A);
								\fill[\monoidallcolour,opacity=\opacity] 
					  				(A) -- (B) -- (C);
					}
	]
		\alias{A} {(IA)B} \&\& \alias{C} {I(AB)} \\
		\& \alias{B} AB
		\arrow["\J\lactalpha", from=1-1, to=1-3]
		\arrow[""{name=0, anchor=center, inner sep=0}, 
							from=1-1, to=2-2] %
		\arrow[""{name=1, anchor=center, inner sep=0}, "\J\lactlambda", from=1-3, to=2-2]
		\arrow[""{name=2, anchor=center, inner sep=0}, "{\lambda \lactprime B}"', curve={height=22pt}, from=1-1, to=2-2]
		\arrow["\oncell{{\theta}}"{description, yshift=0mm}, shorten
	        				<=3pt, shorten >=3pt, draw=none, from=2, to=0]
		\arrow["\oncell{\J{\acttrianglel^{\lact}}}"{description, yshift=0mm, xshift=1mm},
	        shift left=2, shorten <=7pt, shorten >=7pt, draw=none, from=0, to=1]
	\end{tikzcd}
	\hspace{0mm}
	&=
	\hspace{0mm}
	\begin{tikzcd}[
				ampersand replacement = \&,
	    		column sep=1em, scalenodes=1,
	    					execute at end picture={
	    										\foreach \nom in  {A,B,C, D}
	    							  				{\coordinate (\nom) at (\nom.center);}
	    										\fill[\actlcolour,opacity=\opacity] 
	    							  				(A) -- (B) -- (C);
	    							}
	    	]
		\alias{A} {(IA)B} \&\& \alias{B} {I(AB)} \\
		\& \alias{C} AB
		\arrow["\lactalphaprime", from=1-1, to=1-3]
		\arrow[""{name=0, anchor=center, inner sep=0}, "\lactlambdaprime", from=1-3, to=2-2]
		\arrow[""{name=1, anchor=center, inner sep=0}, "{\lambda \lactprime B}"', from=1-1, to=2-2]
		\arrow["\oncell{{\acttrianglel^{\lactprime}}}"{description}, shift left=2, shorten <=7pt, shorten >=7pt, draw=none, from=1, to=0]
	\end{tikzcd}
	\\[1mm]
	\begin{tikzcd}[
		ampersand replacement = \&,	
		row sep = 1.2em,		
	   	column sep=1em, scalenodes=1,
		execute at end picture={
							\foreach \nom in  {A,B,C, D}
				  				{\coordinate (\nom) at (\nom.center);}
							\fill[\monoidalmcolour,opacity=\opacity] 
				  				(A) -- (B) -- (C);
							\fill[\extcolour,opacity=\opacity] 
				  				(A) to[curve={height=24pt}] (C);	
				}
	   	]
		\alias{A} {(AI)B} \&\& \alias{B} {A(IB)} \\
		\& \alias{C} AB
		\arrow["\J\lactalpha", from=1-1, to=1-3]
		\arrow[""{name=0, anchor=center, inner sep=0}, 
						"{\J(A \lact \lactlambda)}"{yshift=1mm}, from=1-3, to=2-2]
		\arrow[""{name=1, anchor=center, inner sep=0}, "{\rho \lactprime B}"', curve={height=22pt}, from=1-1, to=2-2]
		\arrow[""{name=3, anchor=center, inner sep=0}, from=1-1, to=2-2]
		\arrow["\oncell{{\acttrianglem^{\lact}} }"{description, yshift=0mm}, shift left=2, draw=none, from=3, to=0]
		\arrow["\oncell{{\theta}}"{description, yshift=0mm}, draw=none, from=1, to=3]
	\end{tikzcd}
	\hspace{0mm}
	&=
	\hspace{0mm}
	\begin{tikzcd}[
		ampersand replacement = \&,
		row sep = 1.2em,
	   	column sep=1em, scalenodes=1,
		execute at end picture={
							\foreach \nom in  {A,B,C, D}
				  				{\coordinate (\nom) at (\nom.center);}
							\fill[\monoidalmcolour,opacity=\opacity] 
				  				(A) -- (B) -- (C);
							\fill[\extcolour,opacity=\opacity] 
				  				(B) to[curve={height=-24pt}] (C);					  				
				}
	   	]
		\alias{A} {(AI)B} \&\& \alias{B} {A(IB)} \\
		\& \alias{C} AB
		\arrow["\lactalphaprime", from=1-1, to=1-3]
		\arrow[""{name=0, anchor=center, inner sep=0}, from=1-3, to=2-2]
		\arrow[""{name=2, anchor=center, inner sep=0}, "{\J(A \lact \lactlambda)}"{yshift=1mm}, curve={height=-22pt}, from=1-3, to=2-2]
		\arrow[""{name=3, anchor=center, inner sep=0}, 
							"{\rho \lactprime B}"{swap}, from=1-1, to=2-2]
		\arrow["\oncell{{\acttrianglem^{\lactprime}} }"{description, yshift=0mm}, shift left=2, draw=none, from=3, to=0]
		\arrow["\oncell{{\theta}}"{description,yshift=0mm}, draw=none, from=0, to=2]
	\end{tikzcd}
	\end{align*}
\end{minipage}
\hspace{5mm}
	\begin{minipage}{0.3\textwidth}
	\centering
	\[
	\begin{tikzcd}[
			ampersand replacement = \&,
	  		row sep=0.2em, 
	  		column sep=1em, 
	  		scalenodes=1,
		execute at end picture={
							\foreach \nom in  {A,B,C, D,E}
				  				{\coordinate (\nom) at (\nom.center);}
							\fill[\monoidalpcolour,opacity=\opacity] 
				  				(A) -- (B) -- (C) -- (D) -- (E);			 
							\fill[\extcolour,opacity=\opacity] 
				  				(A) to[curve={height=24pt}] (B) -- (A);							  				 				
							\fill[\extcolour,opacity=\opacity] 
				  				(C) to[curve={height=24pt}] (D) -- (C);					  				
		}
	]
	\& \alias{A} {((AB)C)D} \\
	\alias{B} {(A(BC))D} \\
	\&\&[-3em] \alias{E} {(AB)(CD)} \\
	\alias{C} {A((BC)D)} \\
	\& \alias{D} {A(B(CD))}
	\arrow[""{name=0, anchor=center, inner sep=0}, shorten=-0.2em, from=1-2, to=2-1]
	\arrow[""{name=1, anchor=center, inner sep=0}, "\J\lactalpha"', from=2-1, to=4-1]
	\arrow["\J\lactalpha", from=1-2, to=3-3]
	\arrow["\J\lactalpha", from=3-3, to=5-2]
	\arrow[""{name=2, anchor=center, inner sep=0}, "{\alpha \lactprime D}"', curve={height=20pt}, from=1-2, to=2-1]
	\arrow[""{name=3, anchor=center, inner sep=0}, "{A \lactprime \lactalphaprime}"', curve={height=20pt}, from=4-1, to=5-2]
	\arrow[""{name=4, anchor=center, inner sep=0}, shorten=-0.2em, from=4-1, to=5-2]
	\arrow["\oncell{{\J\actpentagonator^{\lact}}}"{description, pos=0.6}, draw=none, from=1, to=3-3]
	\arrow["\oncell{{\theta}}"{description}, draw=none, from=2, to=0]
	\arrow["\oncell{{\theta}}"{description},
       						 draw=none, from=3, to=4, pos=0.46]
	\end{tikzcd}
   \vspace{-2mm}
   \]
	\[
	\vertequals
	\vspace{-4mm}
	\]
   \[
	\begin{tikzcd}[
			ampersand replacement = \&,
	    		scalenodes=1, 
	    		column sep=1em, 
	    		row sep=0em,
		execute at end picture={
							\foreach \nom in  {A,B,C, D,E}
				  				{\coordinate (\nom) at (\nom.center);}
							\fill[\actpcolour,opacity=\opacity] 
				  				(A) -- (B) -- (C) -- (D) -- (E);			 	  				
		}
	]
	\& \alias{A} {((AB)C)D} \\
	\alias{B} {(A(BC))D} \\
	\&\&[-4em] \alias{E} {(AB)(CD)} \\
	\alias{C} {A((BC)D)} \\
	\& \alias{D} {A(B(CD))}
	\arrow[""{name=0, anchor=center, inner sep=0}, "\lactalphaprime"',
	    from=2-1, to=4-1]
	\arrow["\lactalphaprime", from=1-2, to=3-3]
	\arrow["\lactalphaprime", from=3-3, to=5-2]
	\arrow["{\alpha \lactprime D}"', from=1-2, to=2-1, shorten=-0.2em]
	\arrow["{A \lactprime\lactalphaprime}"', from=4-1, to=5-2, shorten=-0.2em] 
	\arrow["\oncell{{\actpentagonator^{\lactprime}}}"{description, xshift=3mm}, draw=none, from=0, to=3-3]
	\end{tikzcd}
   \]
	\end{minipage}
\]

%% file: diag-Freyd-bicat-equations.tex
\begin{subfigure}{\textwidth}
\[
\begin{tikzcd}[
		column sep=1.5em,
		row sep = 1.2em,
		execute at end picture={
							\foreach \nom in  {A,B}
				  				{\coordinate (\nom) at (\nom.center);}
							\fill[\extcolour,opacity=\opacity] 
				  				(A) to[bend left=77] (B) to[bend left=75] (A);
							\fill[\Klambdacolour,opacity=\opacity] 
				  				(A) to[bend left=80] (B) -- (Y) -- (X);
	}
		] 
	\alias{A} {IB} && \alias{X} B \\
	\alias{B} {I B'} && \alias{Y} {B'}
	\arrow["{\J \lambda}", from=1-1, to=1-3]
	\arrow["{\J \lambda}"', from=2-1, to=2-3]
	\arrow[""{name=0, anchor=center, inner sep=0}, "{\J g}", from=1-3, to=2-3]
        \arrow[""{name=1, anchor=center, inner sep=0}, bend left = 40, from=1-1, to=2-1]
	\arrow[""{name=2, anchor=center, inner sep=0}, "{I \ltie \J g}"{swap}', 
						bend right = 40, from=1-1, to=2-1]
	\arrow["\oncell{{\theta}}"{description}, draw=none, from=2, to=1]
	\arrow["\oncell{\J\cellOf{\lambda}_g}"{description}, draw=none, from=1, to=0]
      \end{tikzcd}
      \hspace{1mm}
      =
      \hspace{1mm}
      \begin{tikzcd}[
      	column sep=1.5em,
		row sep = 1.2em,      	
      			execute at end picture={
      								\foreach \nom in  {A,B,C, D}
      					  				{\coordinate (\nom) at (\nom.center);}
      								\fill[\actlambdacolour,opacity=\opacity] 
      					  				(A) -- (B) -- (D) -- (C);
      		}
      ]
	\alias{A} {I B} && \alias{B} B \\
	\alias{C} {I B'} && \alias{D} {B'}
	\arrow["{\lambda}", from=1-1, to=1-3]
	\arrow["{\lambda}"', from=2-1, to=2-3]
	\arrow[""{name=0, anchor=center, inner sep=0}, "{\J g}", from=1-3, to=2-3]
	\arrow[""{name=1, anchor=center, inner sep=0}, "{I \ltie \J g}"', curve={height=0pt}, from=1-1, to=2-1]
	\arrow["\oncell{{\cellOf{\lambda}_{\J g}}}"{description}, draw=none, from=1, to=0]
\end{tikzcd}
\hspace{8mm}
\begin{tikzcd}[
		column sep=1.3em,
		row sep = 1.2em,		
		execute at end picture={
							\foreach \nom in  {A,B}
				  				{\coordinate (\nom) at (\nom.center);}
							\fill[\extcolour,opacity=\opacity] 
				  				(A) to[bend left=77] (B) to[bend left=75] (A);
							\fill[\Klambdacolour,opacity=\opacity] 
				  				(A) to[bend left=80] (B) -- (Y) -- (X);
	}
		] 
	\alias{A} {AI} && \alias{X} A \\
	\alias{B} {A'I} && \alias{Y} {A'}
	\arrow["{\J \rho}", from=1-1, to=1-3]
	\arrow["{\J \rho}"', from=2-1, to=2-3]
	\arrow[""{name=0, anchor=center, inner sep=0}, "{\J f}", from=1-3, to=2-3]
        \arrow[""{name=1, anchor=center, inner sep=0}, bend left = 40, from=1-1, to=2-1]
	\arrow[""{name=2, anchor=center, inner sep=0}, "{\J f \rtie I}"{swap}', 
						bend right = 40, from=1-1, to=2-1]
	\arrow["\oncell{{\theta}}"{description}, draw=none, from=2, to=1]
	\arrow["\oncell{\J\cellOf{\rho}_f}"{description}, draw=none, from=1, to=0]
      \end{tikzcd}
      \hspace{1mm}
      =
      \hspace{1mm}
      \begin{tikzcd}[
      	column sep=1.5em,
		row sep = 1.2em,      	
      			execute at end picture={
      								\foreach \nom in  {A,B,C, D}
      					  				{\coordinate (\nom) at (\nom.center);}
      								\fill[\actlambdacolour,opacity=\opacity] 
      					  				(A) -- (B) -- (D) -- (C);
      		}
      ]
	\alias{A} {AI} && \alias{B} A \\
	\alias{C} {A'I} && \alias{D} {A'}
	\arrow["{\rho}", from=1-1, to=1-3]
	\arrow["{\rho}"', from=2-1, to=2-3]
	\arrow[""{name=0, anchor=center, inner sep=0}, "{\J f}", from=1-3, to=2-3]
	\arrow[""{name=1, anchor=center, inner sep=0}, "{\J f \rtie I}"', curve={height=0pt}, from=1-1, to=2-1]
	\arrow["\oncell{{\cellOf{\rho}_{\J f}}}"{description}, draw=none, from=1, to=0]
\end{tikzcd}
\]
\vspace{-4mm}
\[
\begin{tikzcd}[
		column sep=.8em, scalenodes=1,
		row sep = 1.4em,      			
				execute at end picture={
									\foreach \nom in  {A,B,C, D}
						  				{\coordinate (\nom) at (\nom.center);}
									\fill[\extcolour,opacity=\opacity] 
						  				(A) to[curve={height=-30pt}] (B) to[curve={height=-30pt}] (A);
									\fill[\Klambdacolour,opacity=\opacity] 
						  				(A) to[curve={height=-30pt}] (B) -- (Y) -- (X);
			}
		]
	\alias{A} {(A B) C} && \alias{X} A(B C) \\
	\alias{B} {(A' B) C} && \alias{Y} {A' (B C)}
	\arrow["{\J \alpha}", from=1-1, to=1-3]
	\arrow["{\J \alpha}"', from=2-1, to=2-3]
	\arrow[""{name=0, anchor=center, inner sep=0}, "", from=1-3, to=2-3]
	\arrow["\oncell{\J \cellOf{\alpha}_{f, B, C}}"{description, xshift=5mm}, draw=none, from=2, to=0]
	\arrow[""{name=1, anchor=center, inner sep=0}, 
					curve={height=24pt}, from=1-1, to=2-1]
	\arrow[""{name=2, anchor=center, inner sep=0}, %
        from=1-1, to=2-1]
	\arrow[""{name=3, anchor=center, inner sep=0}, curve={height=0pt}, from=1-1, to=2-1]
	\arrow[""{name=4, anchor=center, inner sep=0}, 
						curve={height=-24pt}, from=1-1, to=2-1] 			
	\arrow["\zeta"{description}, draw=none, from=4, to=3]			
	\arrow["\zeta \rtie C"{yshift=-2mm, xshift=.2mm}, draw=none, from=1, to=3]									
\end{tikzcd}
      \hspace{-2mm}
      =
      \hspace{-2mm}
\begin{tikzcd}[
		column sep=.8em, 
		scalenodes=1,
		row sep = 1.4em,      		
		execute at end picture={
								\foreach \nom in  {A,B,C, D}
					  				{\coordinate (\nom) at (\nom.center);}
								\fill[\extcolour,opacity=\opacity] 
					  				(A) to[curve={height=12pt}] (B) to[curve={height=12pt}]  (A);
								\fill[\Klambdacolour,opacity=\opacity] 
					  				(A) to[curve={height=20pt}]  (B) -- (Y) -- (X);
		}
		]
	\alias{X} {(AB)C} && \alias{A} {A(BC)} \\
	\alias{Y} {(A'B)C} && \alias{B} {A'(BC)}
	\arrow[""{name=0, anchor=center, inner sep=0}, "\alpha", from=1-1, to=1-3]
	\arrow["\alpha"', from=2-1, to=2-3]
	\arrow[""{name=1, anchor=center, inner sep=0}, 
					curve={height=10pt}, from=1-3, to=2-3] 		%
	\arrow[""{name=2, anchor=center, inner sep=0}, %
        from=1-1, to=2-1]
	\arrow[""{name=4, anchor=center, inner sep=0}, "{\J (f(BC))}", 
						curve={height=-10pt}, from=1-3, to=2-3] %
	\arrow["\oncell{{\cellOf{\alpha}_{f, B, C}}}"{description}, draw=none, from=2, to=1]
	\arrow["\oncell{{\zeta}}"{description, xshift=.17mm}, draw=none, from=1, to=4]
\end{tikzcd}
\hspace{3mm}
\begin{tikzcd}[
		column sep=.8em, scalenodes=1,
		row sep = 1.4em,      		
				execute at end picture={
									\foreach \nom in  {A,B,C, D}
						  				{\coordinate (\nom) at (\nom.center);}
									\fill[\extcolour,opacity=\opacity] 
						  				(A) to[curve={height=-30pt}] (B) to[curve={height=-30pt}] (A);
									\fill[\Klambdacolour,opacity=\opacity] 
						  				(A) to[curve={height=-30pt}] (B) -- (Y) -- (X);
			}
		]
	\alias{A} {(A B) C} && \alias{X} A(B C) \\
	\alias{B} {(A B') C} && \alias{Y} {A (B' C)}
	\arrow["{\J \alpha}", from=1-1, to=1-3]
	\arrow["{\J \alpha}"', from=2-1, to=2-3]
	\arrow[""{name=0, anchor=center, inner sep=0}, "", from=1-3, to=2-3]
	\arrow["\oncell{\J \cellOf{\alpha}_{A, g, C}}"{description, xshift=4mm}, draw=none, from=2, to=0]
	\arrow[""{name=1, anchor=center, inner sep=0}, 
					curve={height=24pt}, from=1-1, to=2-1]
	\arrow[""{name=2, anchor=center, inner sep=0}, %
        from=1-1, to=2-1]
	\arrow[""{name=3, anchor=center, inner sep=0}, curve={height=0pt}, from=1-1, to=2-1]
	\arrow[""{name=4, anchor=center, inner sep=0}, 
						curve={height=-24pt}, from=1-1, to=2-1] 			
	\arrow["\zeta"{description}, draw=none, from=4, to=3]			
	\arrow["\theta \rtie C"{yshift=-2mm, xshift=.2mm}, draw=none, from=1, to=3]									
\end{tikzcd}
      \hspace{-2mm}
      =
      \hspace{-2mm}
\begin{tikzcd}[
		column sep=.8em, scalenodes=1,
		row sep = 1.4em,      		
		execute at end picture={
								\foreach \nom in  {A,B,C, D}
					  				{\coordinate (\nom) at (\nom.center);}
								\fill[\extcolour,opacity=\opacity] 
					  				(A) to[curve={height=29pt}] (B) to[curve={height=27pt}]  (A);
								\fill[\Klambdacolour,opacity=\opacity] 
					  				(A) to[curve={height=29pt}]  (B) -- (Y) -- (X);
		}
		]
	\alias{X} {(AB)C} && \alias{A} {A(BC)} \\
	\alias{Y} {(AB)C'} && \alias{B} {A(BC')}
	\arrow[""{name=0, anchor=center, inner sep=0}, "\alpha", from=1-1, to=1-3]
	\arrow["\alpha"', from=2-1, to=2-3]
	\arrow[""{name=1, anchor=center, inner sep=0}, 
					curve={height=22pt}, from=1-3, to=2-3] 		%
	\arrow[""{name=2, anchor=center, inner sep=0}, %
        from=1-1, to=2-1]
	\arrow[""{name=3, anchor=center, inner sep=0}, curve={height=0pt}, from=1-3, to=2-3]
	\arrow[""{name=4, anchor=center, inner sep=0}, "{\J (A(gC))}", 
						curve={height=-20pt}, from=1-3, to=2-3] %
	\arrow["\oncell{{\cellOf{\alpha}_{A, g, C}}}"{description}, draw=none, from=2, to=1]
	\arrow["\oncell{{A \ltie \zeta}}"{yshift=-2mm,xshift=.17mm}, draw=none, from=1, to=3]
	\arrow["\oncell{{\theta}}"{description}, draw=none, from=3, to=4]
\end{tikzcd}
\]
\vspace{-1mm}
\[
	\begin{tikzcd}[
		column sep=2em, scalenodes=1,
		row sep = 1.3em,
				execute at end picture={
									\foreach \nom in  {A,B,C, D}
						  				{\coordinate (\nom) at (\nom.center);}
									\fill[\extcolour,opacity=\opacity] 
						  				(A) to[bend left=77] (B) to[bend left=75] (A);
									\fill[\Klambdacolour,opacity=\opacity] 
						  				(A) to[bend left=80] (B) -- (Y) -- (X);
			}
		]
	\alias{A} {(A B) C} && \alias{X} A(B C) \\
	\alias{B} {(A B) C'} && \alias{Y} {A(B C')}
	\arrow["{\J \alpha}", from=1-1, to=1-3]
	\arrow["{\J \alpha}"', from=2-1, to=2-3]
	\arrow[""{name=0, anchor=center, inner sep=0}, "%
        ", from=1-3, to=2-3]
	\arrow[""{name=1, anchor=center, inner sep=0}, bend left = 40, from=1-1, to=2-1]
	\arrow[""{name=2, anchor=center, inner sep=0}, "{AB \ltie \J h}"', bend right = 40, from=1-1, to=2-1]
	\arrow["\oncell{{\theta}}"{description}, draw=none, from=2, to=1]
	\arrow["\oncell{\J \cellOf{\alpha}_{A, B, h}}"{description}, draw=none, from=1, to=0]
\end{tikzcd}
      \hspace{0mm}
      =
      \hspace{0mm}
\begin{tikzcd}[
		column sep=2em, scalenodes=1,
		row sep = 1.3em,		
		execute at end picture={
								\foreach \nom in  {A,B,C, D}
					  				{\coordinate (\nom) at (\nom.center);}
								\fill[\extcolour,opacity=\opacity] 
					  				(A) to[curve={height=29pt}] (B) to[curve={height=30pt}]  (A);
								\fill[\Klambdacolour,opacity=\opacity] 
					  				(A) to[curve={height=29pt}]  (B) -- (Y) -- (X);
		}
		]
	\alias{X} {(AB)C} && \alias{A} {A(BC)} \\
	\alias{Y} {(AB)C'} && \alias{B} {A(BC')}
	\arrow[""{name=0, anchor=center, inner sep=0}, "\alpha", from=1-1, to=1-3]
	\arrow["\alpha"', from=2-1, to=2-3]
	\arrow[""{name=1, anchor=center, inner sep=0}, 
					curve={height=24pt}, from=1-3, to=2-3] 		%
	\arrow[""{name=2, anchor=center, inner sep=0}, %
        from=1-1, to=2-1]
	\arrow[""{name=3, anchor=center, inner sep=0}, curve={height=0pt}, from=1-3, to=2-3]
	\arrow[""{name=4, anchor=center, inner sep=0}, "{\J (A(Bh))}", 
						curve={height=-24pt}, from=1-3, to=2-3] %
	\arrow["\oncell{{\cellOf{\alpha}_{A, B, h}}}"{description}, draw=none, from=2, to=1]
	\arrow["\oncell{{A \ltie \theta}}"{yshift=-2mm,xshift=.17mm}, draw=none, from=1, to=3]
	\arrow["\oncell{{\theta}}"{description}, draw=none, from=3, to=4]
\end{tikzcd}
\]

\caption{Compatibility rules for structural transformations}
\end{subfigure}

\begin{subfigure}{\textwidth}
\[
\hspace{-25mm}
\begin{minipage}{0.65\textwidth}
\centering
\[
\begin{tikzcd}[
   		column sep=1.5em, scalenodes=1,
		execute at end picture={
							\foreach \nom in  {A,B,C, D}
				  				{\coordinate (\nom) at (\nom.center);}
							\fill[\extcolour,opacity=\opacity] 
				  				(A) to[curve={height=30pt}] (B) -- (A);
							\fill[\monoidallcolour,opacity=\opacity] 
				  				(A) -- (B) -- (C);
				}
]
	\alias{A} {(IA)B} && \alias{C} {I(AB)} \\
	& \alias{B} AB
	\arrow["\J \alpha", from=1-1, to=1-3]
	\arrow[""{name=0, anchor=center, inner sep=0}, "{\scriptscriptstyle \J (\lambda B)}"{description}, 
						from=1-1, to=2-2]
	\arrow[""{name=1, anchor=center, inner sep=0}, "\J \lambda", from=1-3, to=2-2]
	\arrow[""{name=2, anchor=center, inner sep=0}, "{\lambda \rtie B}"', curve={height=28pt}, from=1-1, to=2-2]
	\arrow["\oncell{{\zeta}}"{description, yshift=-1mm}, shorten
        				<=3pt, shorten >=3pt, draw=none, from=2, to=0]
	\arrow["\oncell{{\J \montrianglel}}"{description, yshift=1mm, xshift=1mm},
        shift left=2, shorten <=7pt, shorten >=7pt, draw=none, from=0, to=1]
\end{tikzcd}
\hspace{0mm}
=
\hspace{0mm}
\begin{tikzcd}[
    		column sep=1em, scalenodes=1,
    					execute at end picture={
    										\foreach \nom in  {A,B,C, D}
    							  				{\coordinate (\nom) at (\nom.center);}
    										\fill[\actlcolour,opacity=\opacity] 
    							  				(A) -- (B) -- (C);
    							}
    	]
	\alias{A} {(IA)B} && \alias{B} {I(AB)} \\
	& \alias{C} AB
	\arrow["\alpha", from=1-1, to=1-3]
	\arrow[""{name=0, anchor=center, inner sep=0}, "\lambda", from=1-3, to=2-2]
	\arrow[""{name=1, anchor=center, inner sep=0}, "{\lambda \rtie B}"', from=1-1, to=2-2]
	\arrow["\oncell{{\montrianglel}}"{description}, shift left=2, shorten <=7pt, shorten >=7pt, draw=none, from=1, to=0]
\end{tikzcd}
\]
\vspace{-3mm}
\[
\begin{tikzcd}[
   		column sep=1.5em, scalenodes=1,
		execute at end picture={
							\foreach \nom in  {A,B,C, D}
				  				{\coordinate (\nom) at (\nom.center);}
							\fill[\extcolour,opacity=\opacity] 
				  				(A) -- (C)  to[curve={height=-30pt}] (B);
							\fill[\monoidallcolour,opacity=\opacity] 
				  				(A) -- (B) -- (C);
				}
]
	\alias{A} {(AB)I} && \alias{C} {A(BI)} \\
	& \alias{B} AB
	\arrow["\J \alpha", from=1-1, to=1-3]
	\arrow[""{name=0, anchor=center, inner sep=0}, "\J\rho"{swap}, 
						from=1-1, to=2-2]
	\arrow[""{name=1, anchor=center, inner sep=0}, 
					"{\scriptscriptstyle A \ltie \J \rho}"{description}, from=1-3, to=2-2]
	\arrow[""{name=2, anchor=center, inner sep=0}, "{A \ltie \rho}"{}, curve={height=-28pt}, from=1-3, to=2-2]
	\arrow["\oncell{{\theta}}"{description, yshift=-1mm}, shorten
        				<=3pt, shorten >=3pt, draw=none, from=2, to=1]
	\arrow["\oncell{{\J \montriangler}}"{description, yshift=1mm, xshift=1mm},
        shift left=2, shorten <=7pt, shorten >=7pt, draw=none, from=0, to=1]
\end{tikzcd}
\hspace{0mm}
=
\hspace{0mm}
\begin{tikzcd}[
    		column sep=1em, scalenodes=1,
    					execute at end picture={
    										\foreach \nom in  {A,B,C, D}
    							  				{\coordinate (\nom) at (\nom.center);}
    										\fill[\actlcolour,opacity=\opacity] 
    							  				(A) -- (B) -- (C);
    							}
    	]
	\alias{A} {(AB)I} && \alias{B} {A(BI)} \\
	& \alias{C} AB
	\arrow["\alpha", from=1-1, to=1-3]
	\arrow[""{name=0, anchor=center, inner sep=0}, "A \ltie \rho", from=1-3, to=2-2]
	\arrow[""{name=1, anchor=center, inner sep=0}, "{\rho}"', from=1-1, to=2-2]
	\arrow["\oncell{{\montriangler}}"{description}, shift left=2, shorten <=7pt, shorten >=7pt, draw=none, from=1, to=0]
\end{tikzcd}
\]
\vspace{-3mm}
\[
\begin{tikzcd}[
   	column sep=1.5em, scalenodes=1,
	execute at end picture={
						\foreach \nom in  {A,B,C, D}
			  				{\coordinate (\nom) at (\nom.center);}
						\fill[\monoidalmcolour,opacity=\opacity] 
			  				(A) -- (B) -- (C);
						\fill[\extcolour,opacity=\opacity] 
			  				(A) to[curve={height=33pt}] (C);	
						\fill[\extcolour,opacity=\opacity] 
			  				(B) to[curve={height=-33pt}] (C);					  				
			}
   	]
	\alias{A} {(AI)B} && \alias{B} {A(IB)} \\
	& \alias{C} AB
	\arrow["\J \alpha", from=1-1, to=1-3]
	\arrow[""{name=0, anchor=center, inner sep=0}, 
					"{\scriptscriptstyle \J (A\lambda)}"{description}, from=1-3, to=2-2]
	\arrow[""{name=1, anchor=center, inner sep=0}, "{\rho \rtie B}"', curve={height=30pt}, from=1-1, to=2-2]
	\arrow[""{name=2, anchor=center, inner sep=0}, "{A \ltie \lambda}", curve={height=-30pt}, from=1-3, to=2-2]
	\arrow[""{name=3, anchor=center, inner sep=0}, 
						"{\scriptscriptstyle \J (\rho B)}"{description}, from=1-1, to=2-2]
	\arrow["\oncell{{\J \montrianglem} }"{description, yshift=1mm}, shift left=2, draw=none, from=3, to=0]
	\arrow["\oncell{{\theta}}"{description,yshift=-1mm}, draw=none, from=0, to=2]
	\arrow["\oncell{{\zeta}}"{description, yshift=-1mm}, draw=none, from=1, to=3]
\end{tikzcd}
\hspace{0mm}
=
\hspace{0mm}
\begin{tikzcd}[
   	column sep=1em, scalenodes=1,
	execute at end picture={
						\foreach \nom in  {A,B,C, D}
			  				{\coordinate (\nom) at (\nom.center);}
						\fill[\actmcolour,opacity=\opacity] 
			  				(A) -- (B) -- (C);			  				
			}
]
	\alias{A} {(AI)B} && \alias{B} {A(IB)} \\
	& \alias{C} AB
	\arrow["\alpha", from=1-1, to=1-3]
	\arrow[""{name=0, anchor=center, inner sep=0}, "{A \ltie \lambda}", from=1-3, to=2-2]
	\arrow[""{name=1, anchor=center, inner sep=0}, "{\rho \rtie B}"', from=1-1, to=2-2]
	\arrow["\oncell{{\montrianglem} }"{description}, shift left=2, draw=none, from=1, to=0]
\end{tikzcd}
\]
\end{minipage}
\hspace{2mm}
\begin{minipage}{0.18\textwidth}
\centering
\begin{tikzcd}[
   		row sep=-.5em, 
   		column sep=1em, 
   		scalenodes=1,
		execute at end picture={
							\foreach \nom in  {A,B,C, D,E}
				  				{\coordinate (\nom) at (\nom.center);}
							\fill[\monoidalpcolour,opacity=\opacity] 
				  				(A) -- (B) -- (C) -- (D) -- (E);			 
							\fill[\extcolour,opacity=\opacity] 
				  				(A) to[curve={height=24pt}] (B) -- (A);							  				 				
							\fill[\extcolour,opacity=\opacity] 
				  				(C) to[curve={height=24pt}] (D) -- (C);					  				
		}
]
	& \alias{A} {((AB)C)D} \\
	\alias{B} {(A(BC))D} \\
	&&[-3em] \alias{E} {(AB)(CD)} \\
	\alias{C} {A((BC)D)} \\
	& \alias{D} {A(B(CD))}
	\arrow[""{name=0, anchor=center, inner sep=0}, shorten=-0.2em, from=1-2, to=2-1]
	\arrow[""{name=1, anchor=center, inner sep=0}, "\J \alpha"', from=2-1, to=4-1]
	\arrow["\J \alpha", from=1-2, to=3-3]
	\arrow["\J \alpha", from=3-3, to=5-2]
	\arrow[""{name=2, anchor=center, inner sep=0}, "{\alpha \rtie D}"'{yshift=-1mm}, curve={height=20pt}, from=1-2, to=2-1]
	\arrow[""{name=3, anchor=center, inner sep=0}, "{A \ltie \alpha}"'{yshift=1mm}, curve={height=20pt}, from=4-1, to=5-2]
	\arrow[""{name=4, anchor=center, inner sep=0}, shorten=-0.2em, from=4-1, to=5-2]
	\arrow["\oncell{{\J \pentagonator}}"{description, pos=0.6}, draw=none, from=1, to=3-3]
	\arrow["\oncell{{\zeta}}"{description}, draw=none, from=2, to=0]
	\arrow["\oncell{{\theta}}"{description},
       						 draw=none, from=3, to=4, pos=0.46]
\end{tikzcd}
\\
\vspace{-4mm}
\[ \hspace{20mm}\vertequals \]
\vspace{-8mm}
\begin{tikzcd}[
     		scalenodes=1, 
     		column sep=1em, 
     		row sep=-.3em,
		execute at end picture={
							\foreach \nom in  {A,B,C, D,E}
				  				{\coordinate (\nom) at (\nom.center);}
							\fill[\actpcolour,opacity=\opacity] 
				  				(A) -- (B) -- (C) -- (D) -- (E);			 	  				
		}
  ]
	& \alias{A} {((AB)C)D} \\
	\alias{B} {(A(BC))D} \\
	&&[-4em] \alias{E} {(AB)(CD)} \\
	\alias{C} {A((BC)D)} \\
	& \alias{D} {A(B(CD))}
	\arrow[""{name=0, anchor=center, inner sep=0}, "\alpha"',
        from=2-1, to=4-1]
	\arrow["\alpha", from=1-2, to=3-3]
	\arrow["\alpha", from=3-3, to=5-2]
	\arrow["{\alpha \rtie D}"', from=1-2, to=2-1, shorten=-0.2em]
	\arrow["{A \ltie \alpha}"', from=4-1, to=5-2, shorten=-0.2em] 
	\arrow["\oncell{{\pentagonator}}"{description, xshift=3mm}, draw=none, from=0, to=3-3]
\end{tikzcd}
\end{minipage}
\]
\caption{Compatibility rules for structural modifications}

\end{subfigure}

%% file: sec-conclusion.tex
\paragraph{Summary.} 
We have introduced bicategorical versions of premonoidal categories 
	(\Cref{def:premonoidal-bicategory})
and Freyd categories 
	(\Cref{def:freyd-bicategory}). 
Along the way we have observed subtleties that arise only in the
2-dimensional setting, and discussed simple canonical
examples. Finally, we have connected our theory to the existing literature by showing our definition is equivalent to certain actions in the expected way.

This paper develops abstract categorical notions, but these are
intended to be immediately practical. Specifically, the literature contains no satisfying account of
call-by-value languages in bicategories of games
(\cite{template-games,cg1}), spans \cite{Fiadeiro2007}, or
profunctors \cite{FioreSpecies}, and this work offers a technical basis to fill that
gap. Our next steps will be in this direction.

\paragraph*{Perspectives.}
This work takes place in a broader line of research on bicategorical
semantic structures, and there are several avenues to explore. We expect a tight connection between Freyd
bicategories and recently-developed notions of strength for
pseudomonads on monoidal bicategories (\cite{TanakaThesis,DBLP:journals/corr/abs-2304-11014,Slattery2023}).
Freyd bicategories should also be related to a 2-dimensional notion of
\emph{arrows}, based on $\Cat$-valued profunctors, yet to be
developed (\cite{DBLP:journals/entcs/HeunenJ06,corner2016day}). 

In particular,
the Kleisli bicategory of a strong pseudomonad should be premonoidal,
and the canonical functor from the base category should give a Freyd
structure, and conversely, a \emph{closed} Freyd bicategory should be equivalent
to a strong pseudomonad together with Kleisli exponentials. 
From a syntactic perspective, we expect cartesian Freyd bicategories to have an internal language similar to fine-grained call-by-value $\lambda$-calculus~\cite{Levy2003}, with the addition of \emph{rewrites} between terms
	(\cf~\cite{Seely1987,Hilken1996,Hirschowitz2013,LICS2019}).

In a more theoretical direction, although the centre of a premonoidal
category is always a monoidal category, this does not happen in the
bicategorical setting. Roughly speaking, for central $f$ and $g$, the
interchange of $f$ and $g$ is witnessed independently by 2-cells $\lc{f}_g$ and
$(\rc{g}_f)^{-1}$. This leads to ambiguity and it is not clear
how to define the pseudofunctor $\tens$; indeed, it is not even clear that these 2-cells are themselves central. In this paper we have shown that
the centre is a binoidal bicategory, and in further work we will give
a more complete description of its structure, along with
an alternative presentation of Freyd bicategories in terms of centrality
witnesses.

%% file: diag-extra-diagrams-for-def-of-Freyd-action.tex
\vspace{-5mm}
\[
\hspace{-1mm}
\begin{minipage}{0.5\textwidth}
\[
\begin{tikzcd}[
		column sep = 1.3em, 
		scalenodes = .8,
		execute at end picture={
							\foreach \nom in  {A,B,C, D, X}
				  				{\coordinate (\nom) at (\nom.center);}
							\fill[\extcolour,opacity=\opacity] 
				  				(A) to[curve={height=9pt}] (B) to[curve={height=9pt}] (A);	
							\fill[\monoidallcolour,opacity=\opacity] 
				  				(A) to[curve={height=9pt}] (B) to (C) to (A);		  		
							\fill[\kcolour,opacity=\opacity] 
				  				(A) to (C) to (D) to (X);		  						  						  					  				
				}
		]
\alias{A} {(IA)B} && \alias{B} AB \\
\alias{X} {(IA')B} & \alias{C} {I(AB)} & {A'B} \\
& \alias{D} {I(A'B)}
\arrow["\alpha"'{description}, from=1-1, to=2-2]
\arrow["\lambda"'{description},, from=2-2, to=1-3]
\arrow[""{name=0, anchor=center, inner sep=0}, from=2-2, to=3-2]
\arrow[""{name=1, anchor=center, inner sep=0},  from=1-1, to=2-1]
\arrow["\alpha"', from=2-1, to=3-2]
\arrow["\lambda"', from=3-2, to=2-3]
\arrow[""{name=2, anchor=center, inner sep=0}, from=1-3, to=2-3, "{a \ract B}"]
\arrow[""{name=3, anchor=center, inner sep=0}, "{\lambda \ract B}", curve={height=-8pt}, from=1-1, to=1-3]
\arrow[""{name=4, anchor=center, inner sep=0}, curve={height=8pt}, from=1-1, to=1-3]
\arrow["\acttrianglel^{\lact}"{description,yshift=0mm}, draw=none, from=4, to=2-2]
\arrow["\nu"{description}, shorten <=2pt, shorten >=2pt, Rightarrow, from=3, to=4]
\arrow["{\cellOf{\kappa}_{I, a, B}}"{description}, draw=none, from=1, to=0]
\arrow["{\cellOf{\lactlambda}}"{description}, draw=none, from=2, to=0]
\end{tikzcd}
\hspace{-1mm}
=
\hspace{-3mm}
\begin{tikzcd}[
	column sep = -.1em, 
	scalenodes = .8,
	execute at end picture={
						\foreach \nom in  {A,B,C, D, X}
			  				{\coordinate (\nom) at (\nom.center);}
						\fill[\extcolour,opacity=\opacity] 
			  				(A) to[curve={height=8pt}] (B) to[curve={height=9pt}] (A);	
						\fill[\monoidallcolour,opacity=\opacity] 
			  				(A) to[curve={height=8pt}] (B) to (C) to (A);		  			  						  					  				
			}
]
	{(IA)B} && AB \\
	\alias{A} {(IA')B} && \alias{B} {A'B} \\
	& \alias{C} {I(A'B)}
	\arrow[""{name=0, anchor=center, inner sep=0}, from=1-1, to=2-1, "{(I \lact a) \ract B}"'] 
	\arrow[""{name=1, anchor=center, inner sep=0}, "\alpha"', from=2-1, to=3-2]
	\arrow[""{name=2, anchor=center, inner sep=0}, "\lambda"', from=3-2, to=2-3]
	\arrow[""{name=3, anchor=center, inner sep=0}, from=1-3, to=2-3]
	\arrow["{\lambda \ract B}", curve={height=-6pt}, from=1-1, to=1-3]
	\arrow[""{name=4, anchor=center, inner sep=0}, "{\lambda \ract B}", curve={height=-6pt}, from=2-1, to=2-3]
	\arrow[""{name=5, anchor=center, inner sep=0}, curve={height=6pt}, from=2-1, to=2-3]
	\arrow["{\cellOf{\lactlambda} \ract B}"{description, yshift=3mm}, draw=none, from=0, to=3]
	\arrow["{\acttrianglel^{\lact}}"{description, yshift=1mm}, draw=none, from=1, to=2]
	\arrow["\nu"{description}, shorten <=2pt, shorten >=2pt, Rightarrow, from=4, to=5]
\end{tikzcd}
\]
\end{minipage}
\hspace{2mm}
\begin{minipage}{0.5\textwidth}
\[
\begin{tikzcd}[
	column sep = 1.45em, 
	scalenodes = .8,
	execute at end picture={
						\foreach \nom in  {A,B,C, D, X, Y, Z}
			  				{\coordinate (\nom) at (\nom.center);}
						\fill[\extcolour,opacity=\opacity] 
			  				(A) to[curve={height=9pt}] (B) to[curve={height=11pt}] (A);	
						\fill[\actlcolour,opacity=\opacity] 
			  				(X) to[curve={height=-11pt}] (B) to[curve={height=11pt}] (A) to (X);						
						\fill[\kcolour,opacity=\opacity] 
			  				(X) to (A) to (Z) to (Y);		  						  						  					  				
			}	
]
	\alias{X} {(AB)I} && \alias{B} AB \\
	\alias{Y} {(AB')I} & \alias{A} {A(BI)} & {AB'} \\
	& \alias{Z} {A(B'I)}
	\arrow[""{name=0, anchor=center, inner sep=0}, from=1-1, to=2-1]
	\arrow["\alpha"', from=2-1, to=3-2]
	\arrow["{A \lact \rho}"', curve={height=6pt}, from=3-2, to=2-3]
	\arrow[""{name=1, anchor=center, inner sep=0}, from=1-3, to=2-3, "{A \lact b}"] 
	\arrow[""{name=2, anchor=center, inner sep=0}, "\rho", curve={height=-10pt}, from=1-1, to=1-3]
	\arrow["\alpha"{description}, from=1-1, to=2-2]
	\arrow[""{name=3, anchor=center, inner sep=0}, curve={height=6pt}, from=2-2, to=1-3]
	\arrow[""{name=4, anchor=center, inner sep=0}, from=2-2, to=3-2]
	\arrow[""{name=5, anchor=center, inner sep=0}, curve={height=-10pt}, from=2-2, to=1-3]
	\arrow["{\acttriangler^{\ract}}"{description}, draw=none, from=2, to=2-2]
	\arrow["\nu"{description}, draw=none, from=5, to=3]
	\arrow["{\cellOf{\kappa}_{A, b, I}}"{description, xshift=.5mm, yshift=-.5mm}, draw=none, from=0, to=4]
	\arrow["{A \lact \cellOf{\ractlambda}}"{description, xshift = -.5mm, yshift = .5mm}, curve={height=6pt}, draw=none, from=4, to=1]
\end{tikzcd}
\hspace{-1mm}
=
\hspace{-3mm}
\begin{tikzcd}[
	column sep = -.1em, 
	scalenodes = .8,
	execute at end picture={
						\foreach \nom in  {A,B,C, D, X}
			  				{\coordinate (\nom) at (\nom.center);}
						\fill[\extcolour,opacity=\opacity] 
			  				(A) to[curve={height=8pt}] (B) to[curve={height=10pt}] (A);	
						\fill[\actlcolour,opacity=\opacity] 
			  				(A) to[curve={height=-10pt}] (B) to[curve={height=11pt}] (C) to (A);		  			  						  					  				
			}	
]
	{(AB)I} && AB \\
	\alias{C} {(AB')I} && \alias{B} {AB'} \\
	& \alias{A} {A(B'I)}
	\arrow[""{name=0, anchor=center, inner sep=0}, from=1-1, to=2-1, "{(A \lact b) \ract I}"'] 
	\arrow["\alpha"', from=2-1, to=3-2]
	\arrow[""{name=1, anchor=center, inner sep=0}, "{A \lact \rho}"', curve={height=6pt}, from=3-2, to=2-3]
	\arrow[""{name=2, anchor=center, inner sep=0},from=1-3, to=2-3]
	\arrow["\rho", curve={height=-6pt}, from=1-1, to=1-3]
	\arrow[""{name=3, anchor=center, inner sep=0}, "\rho"{description}, curve={height=-9pt}, from=2-1, to=2-3]
	\arrow[""{name=4, anchor=center, inner sep=0}, curve={height=-9pt}, from=3-2, to=2-3]	%
	\arrow["\nu"{description}, draw=none, from=4, to=1]
	\arrow["{\acttriangler^{\ract}}"{description, yshift=2mm}, draw=none, from=3, to=3-2]
	\arrow["{\cellOf{\ractlambda}}"{description, yshift=2mm}, draw=none, from=0, to=2]
\end{tikzcd}
\]
\end{minipage}
\]
\vspace{-3mm}
\[
\begin{tikzcd}[
	row sep = 1em, 
	column sep = 3em, 
	scalenodes = .9,
	execute at end picture={
						\foreach \nom in  {A,B,C, D, X, Y, Z, W}
			  				{\coordinate (\nom) at (\nom.center);}
						\fill[\extcolour,opacity=\opacity] 
			  				(A) to[curve={height=10pt}] (B) to[curve={height=10pt}] (A);	
						\fill[\actpcolour,opacity=\opacity] 
			  				(X) to (B) to[curve={height=10pt}] (A) to (Y) to (X);	
						\fill[\kcolour,opacity=\opacity] 
							(Z) to (X) to (Y) to (A) to (W) to (Z);		  				  			  						  					  				
			}		
]
	\alias{Z} {((AB)C)D} && \alias{X} {((AB')C)D} \\
	{(A(BC)))D} & \alias{Y} {(A(B'C))D} \\
	\alias{W} {A((BC)D)} & \alias{A} {A((B'C)D)} \\
	{A(B(CD))} && \alias{B} {A(B'(CD))}
	\arrow["{((A \lact b) \ract C) \ract D}", from=1-1, to=1-3]
	\arrow[""{name=0, anchor=center, inner sep=0}, "{\alpha \ract D}"{description}, from=1-3, to=2-2]
	\arrow[""{name=1, anchor=center, inner sep=0}, from=2-2, to=3-2]
	\arrow[""{name=2, anchor=center, inner sep=0}, "{A \ract \alpha}", curve={height=-8pt}, from=3-2, to=4-3]
	\arrow[""{name=3, anchor=center, inner sep=0}, "{\alpha \circ\alpha}", from=1-3, to=4-3]
	\arrow[""{name=4, anchor=center, inner sep=0}, curve={height=8pt}, from=3-2, to=4-3]
	\arrow[""{name=5, anchor=center, inner sep=0}, "{\alpha \ract D}"', from=1-1, to=2-1]
	\arrow[""{name=6, anchor=center, inner sep=0}, "\alpha"', from=2-1, to=3-1]
	\arrow[""{name=7, anchor=center, inner sep=0}, "{A \lact \alpha}"', from=3-1, to=4-1]
	\arrow["{A \lact (b \ract CD)}"', from=4-1, to=4-3]
	\arrow[from=3-1, to=3-2]
	\arrow[from=2-1, to=2-2]
	\arrow["\nu"{description}, draw=none, from=2, to=4]
	\arrow["{\actpentagonator^{\ract}}"{description, xshift=2mm, yshift=1mm}, draw=none, from=1, to=3]
	\arrow["{\cellOf{\kappa}_{A, b, C} \ract D}"{description}, draw=none, from=5, to=0]
	\arrow["{\cellOf{\kappa}_{A, b \ract C, D}}"{description}, draw=none, from=6, to=1]
	\arrow["{A \lact \cellOf{\ractalpha}}"{description, xshift=-1mm, yshift =1mm}, draw=none, from=7, to=4]
\end{tikzcd}
\hspace{1mm}
=
\hspace{1mm}
\begin{tikzcd}[
		row sep = 1em, 
		column sep = 3em, 
		scalenodes = .9,
		execute at end picture={
							\foreach \nom in  {A,B,C, D, X, Y, Z, U}
				  				{\coordinate (\nom) at (\nom.center);}
							\fill[\extcolour,opacity=\opacity] 
				  				(A) to[curve={height=8pt}] (B) to[curve={height=8pt}] (A);	
							\fill[\actpcolour,opacity=\opacity] 
				  				(U) to (Z) to (B) to[curve={height=8pt}] (A) to (U);
							\fill[\kcolour,opacity=\opacity] 
								(Z) to (X) to (Y) to (B) to (Z);		  				  			  						  					  				
				}		
	]
	\alias{U} {((AB)C)D} && {((AB')C)D} \\
	{(A(BC)))D} & \alias{Z} {(AB)(CD)} & \alias{X} {(AB')(CD)} \\
	\alias{A} {A((BC)D)} \\
	\alias{B} {A(B(CD))} && \alias{Y} {A(B'(CD))}
	\arrow["{((A \lact b) \ract C) \ract D}", from=1-1, to=1-3]
	\arrow["{\alpha \ract D}"', from=1-1, to=2-1]
	\arrow[""{name=0, anchor=center, inner sep=0}, "\alpha"', from=2-1, to=3-1]
	\arrow[""{name=1, anchor=center, inner sep=0}, "{A \lact \alpha}"', curve={height=6pt}, from=3-1, to=4-1]
	\arrow["{A \lact (b \ract CD)}"', from=4-1, to=4-3]
	\arrow[""{name=2, anchor=center, inner sep=0}, "\alpha", from=1-3, to=2-3]
	\arrow[""{name=3, anchor=center, inner sep=0}, "\alpha", from=2-3, to=4-3]
	\arrow[""{name=4, anchor=center, inner sep=0}, "\alpha"{description}, from=1-1, to=2-2]
	\arrow[""{name=5, anchor=center, inner sep=0}, "\alpha"{description}, from=2-2, to=4-1]
	\arrow[from=2-2, to=2-3]
	\arrow[""{name=6, anchor=center, inner sep=0}, curve={height=-6pt}, from=3-1, to=4-1]
	\arrow["\nu"{description}, draw=none, from=1, to=6]
	\arrow["{\cellOf{\kappa}_{A, b, CD}}"{description}, draw=none, from=5, to=3]
	\arrow["{\cellOf{\ractalpha}}"{description}, draw=none, from=4, to=2]
	\arrow["{\actpentagonator^{\ract}}"{description, xshift = 2mm}, draw=none, from=0, to=2-2]
\end{tikzcd}
\]
\vspace{-2mm}
\[
\begin{tikzcd}[
		row sep = 1em, 
		column sep = 3em, 
		scalenodes = .9,
		execute at end picture={
							\foreach \nom in  {A,B,C, D, X, Y, Z, W, F}
				  				{\coordinate (\nom) at (\nom.center);}
							\fill[\extcolour,opacity=\opacity] 
				  				(A) to[curve={height=10pt}] (B) to[curve={height=10pt}] (A);	
							\fill[\actpcolour,opacity=\opacity] 
				  				(X) to (B) to[curve={height=10pt}] (A) to (Y) to (X);	
							\fill[\kcolour,opacity=\opacity] 
								(W) to (Y) to (A) to[curve={height=10pt}] (B) to (F) to (W);		  				  			  						  					  				
				}		
	]
	\alias{Z} {((AB)C)D} && \alias{X} {((AB)C')D} \\
	\alias{W} {(A(BC)))D} & \alias{Y} {(A(BC'))D} \\
	{A((BC)D)} & \alias{A} {A((BC')D)} \\
	\alias{F} {A(B(CD))} && \alias{B} {A(B(C'D))}
	\arrow["{(AB \lact c) \ract D}", from=1-1, to=1-3]
	\arrow[""{name=0, anchor=center, inner sep=0}, "{\alpha \ract D}"{description}, from=1-3, to=2-2]
	\arrow[""{name=1, anchor=center, inner sep=0}, from=2-2, to=3-2]
	\arrow[""{name=2, anchor=center, inner sep=0}, "{A \ract \alpha}", curve={height=-8pt}, from=3-2, to=4-3]
	\arrow[""{name=3, anchor=center, inner sep=0}, "{\alpha \circ\alpha}", from=1-3, to=4-3]
	\arrow[""{name=4, anchor=center, inner sep=0}, curve={height=8pt}, from=3-2, to=4-3]
	\arrow[""{name=5, anchor=center, inner sep=0}, "{\alpha \ract D}"', from=1-1, to=2-1]
	\arrow[""{name=6, anchor=center, inner sep=0}, "\alpha"', from=2-1, to=3-1]
	\arrow[""{name=7, anchor=center, inner sep=0}, "{A \lact \alpha}"', from=3-1, to=4-1]
	\arrow["{A \lact (B \lact (c \ract D)}"', from=4-1, to=4-3]
	\arrow[from=3-1, to=3-2]
	\arrow[from=2-1, to=2-2]
	\arrow["\nu"{description}, draw=none, from=2, to=4]
	\arrow["{\actpentagonator^{\ract}}"{description}, draw=none, from=1, to=3]
	\arrow["{\cellOf{\lactalpha} \ract D}"{description}, draw=none, from=5, to=0]
	\arrow["{\cellOf{\kappa}_{A, B \lact c, D}}"{description}, draw=none, from=6, to=1]
	\arrow["{A \lact \cellOf{\kappa}_{B, c, D}}"{description}, draw=none, from=7, to=4]
\end{tikzcd}
\hspace{1mm}
=
\hspace{1mm}
\begin{tikzcd}[
		row sep = 1em, 
		column sep = 3em, 
		scalenodes = .9,
		execute at end picture={
							\foreach \nom in  {A,B,C, D, X, Y, Z, U, F}
				  				{\coordinate (\nom) at (\nom.center);}
							\fill[\extcolour,opacity=\opacity] 
				  				(A) to[curve={height=8pt}] (B) to[curve={height=8pt}] (A);	
							\fill[\actpcolour,opacity=\opacity] 
				  				(U) to (Z) to (B) to[curve={height=8pt}] (A) to (U);
							\fill[\kcolour,opacity=\opacity] 
								(F) to (D) to (X) to (Z);		  				  			  						  					  				
				}				
	]
	\alias{F} {((AB)C)D} && \alias{D} {((AB)C')D} \\
	{(A(BC)))D} & \alias{Z} {(AB)(CD)} & \alias{X} {(AB)(C'D)} \\
	\alias{A} {A((BC)D)} \\
	\alias{B} {A(B(CD))} && \alias{Y} {A(B(C'D))}
	\arrow["{((AB \lact c) \ract D}", from=1-1, to=1-3]
	\arrow["{\alpha \ract D}"', from=1-1, to=2-1]
	\arrow[""{name=0, anchor=center, inner sep=0}, "\alpha"', from=2-1, to=3-1]
	\arrow[""{name=1, anchor=center, inner sep=0}, "{A \lact \alpha}"', curve={height=6pt}, from=3-1, to=4-1]
	\arrow["{A \lact (B \lact (c \ract D)}"', from=4-1, to=4-3]
	\arrow[""{name=2, anchor=center, inner sep=0}, "\alpha", from=1-3, to=2-3]
	\arrow[""{name=3, anchor=center, inner sep=0}, "\alpha", from=2-3, to=4-3]
	\arrow[""{name=4, anchor=center, inner sep=0}, "\alpha"{description}, from=1-1, to=2-2]
	\arrow[""{name=5, anchor=center, inner sep=0}, "\alpha"{description}, from=2-2, to=4-1]
	\arrow[from=2-2, to=2-3]
	\arrow[""{name=6, anchor=center, inner sep=0}, curve={height=-6pt}, from=3-1, to=4-1]
	\arrow["\nu"{description}, draw=none, from=1, to=6]
	\arrow["{\cellOf{\lactalpha}}"{description}, draw=none, from=5, to=3]
	\arrow["{\cellOf{\kappa}_{AB, c, D}}"{description}, draw=none, from=4, to=2]
	\arrow["{\actpentagonator^{\ract}}"{description, xshift = 2mm}, draw=none, from=0, to=2-2]
\end{tikzcd}
\]

%% file: sec-action-proofs.tex
\subsection{From Freyd action to Freyd bicategory}

Fix a Freyd action 
	$(\lact, \theta, \ract, \zeta, \kappa)$
over 
	$\J : \V \to \B$.
We construct a Freyd bicategory with the same underlying pseudofunctor.
For the unit of the premonoidal structure we take the unit $\tensu$ for $\V$. 
Next define
	$A \ltie (-) := A \lact (-)$
and
	$(-) \rtie B := (-) \ract B$. 
The icons $\theta$ and $\zeta$ for the Freyd action then determine the required icons component-wise: 
\[
	 A \ltie J(-) =  A \lact (-) \XRA{\theta_{A, -}} \J(A \tens -) 
	 \qquad,\qquad
	 \J(-) \rtie B  =  J(-) \ract B \XRA{\zeta_{-, B}} \J(- \tens B).
\]
The left- and right unitors are given by 
	$\actlambda^{\lact} : IA \to A$
and 
	$\actlambda^{\ract} : AI \to A$
respectively, and the associator by $\J(\alpha)$ with 2-cell components given by the witnessing 2-cells for
	$\actalpha^{\lact}$,
	$\kappa$,
and
	$\actalpha^{\ract}$. 
The compatibility laws of a Freyd action immediately give the compatibility laws of a Freyd bicategory.  
Similarly, the structural modifications are wholly determined by the definition of a Freyd bicategory: for example, the pentagonator $\pentagonator$ in $\B$ is $\J(\pentagonator)$ composed with $\theta$ and $\zeta$ as in \Cref{def:freyd-bicategory}.
The axioms of a premonoidal bicategory are then checked using the compatibility laws and the corresponding axioms in~$\V$.

It remains to show that the 2-cell components of $\theta$ and $\zeta$ are central and that $\J$ factors through the centre. The former is a short direct calculation. For the latter, for 
	$f : X \to X'$ 
in $\V$ and
	$a : A \to A'$
in $\B$ we define $\lc{f}_a$  
using $\theta, \zeta$ and the interchange laws for the pseudofunctors underlying the actions as in the diagram to the right; $\rc{f}_a$ is similar. 
We write 
$\nu$ for the composite 
	$\J f \ract g \XRA{\zeta} \J(f \tens g) \XRA{\theta} f \lact \J g$.
\begin{wrapfigure}[8]{r}{5cm}
\vspace{-8mm}
\[
\lc{\J f}_a :=
\begin{tikzcd}[
	row sep = 4em, 
	column sep = 5em,
	execute at end picture={
						\foreach \nom in  {A,B,C, D, X, Y}
			  				{\coordinate (\nom) at (\nom.center);}
						\fill[\extcolour,opacity=\opacity] 
			  				(A) to[curve={height=14pt}] (B) to[curve={height=14pt}] (A);
						\fill[\extcolour,opacity=\opacity] 
			  				(X) to[curve={height=14pt}] (Y) to[curve={height=14pt}] (X);			  				
	}   	
	]
	\alias{A} XA & \alias{B} {X'A} \\
	\alias{X} {XA'} & \alias{Y} {X'A'}
	\arrow[""{name=0, anchor=center, inner sep=0}, "{X \lact a}"{description}', from=1-1, to=2-1]
	\arrow[""{name=2, anchor=center, inner sep=0}, "{X' \lact a}"{description}, from=1-2, to=2-2]
	\arrow["{f \lact A'}"{name = 3}, curve={height=-12pt}, from=2-1, to=2-2]
	\arrow["{\J(f) \ract A}"{name = 4}, curve={height=-12pt}, from=1-1, to=1-2]
	\arrow[""{name=5, anchor=center, inner sep=0}, "{f \lact A}"{swap}, curve={height=12pt}, from=1-1, to=1-2]
	\arrow[""{name=7, anchor=center, inner sep=0}, "{\J(f) \ract A'}"', curve={height=12pt}, from=2-1, to=2-2]
	\arrow["\cong"{description}, draw=none, from=0, to=2]
	\arrow["\nu"{description}, draw=none, from=4, to=5]
	\arrow["\nu"{description}, draw=none, from=3, to=7]
\end{tikzcd}
\]
\end{wrapfigure}
Thus, we define $\J'(f) := (f, \lc{\J f}, \rc{\J f})$. 
For any 2-cell $\sigma : f \To f'$ in $\V$, we get that $\J(\sigma)$ is natural by naturality of all the data defining $\lc{\J f}$ and $\rc{\J f}$. 
Finally, one shows that the unit and compositor for $\J$ are central using the identity and composition laws of the icons $\theta$ and $\zeta$. 

In summary, we have the following:

\begin{proposition} %
\label{res:Freyd-action-to-Freyd-bicategory}
Every Freyd action with underlying pseudofunctor
	$\J : \V \to \B$
determines a Freyd bicategory with the same underlying pseudofunctor. 
\end{proposition}

\subsection{From Freyd bicategory to Freyd action}
Let $\freydCat = (\V \xra{\J} \B)$ be a Freyd bicategory. 
First we shall show how to construct a left action 
	$\lact : \V \times \B \to B$; 
the right action is constructed similarly. 
Thereafter we shall show how to construct the rest of the data for a Freyd action.

\paragraph*{From Freyd bicategory to a left action.}
\begin{wrapfigure}[7]{r}{6cm}
\vspace{-9mm}
\[
\begin{tikzcd}[row sep = 1.8em, column sep = 3em]
	XB \\
	{X'B} & {X''B} \\
	{X'B'} & {X''B'} & {X''B''}
	\arrow["{\J(f) \rtie B}"', from=1-1, to=2-1]
	\arrow[""{name=0, anchor=center, inner sep=0}, "{X' \rtie b}"', from=2-1, to=3-1]
	\arrow["{\J(f') \rtie B'}"', from=3-1, to=3-2]
	\arrow["{X'' \ltie b}"', from=3-2, to=3-3]
	\arrow["{\J(f') \rtie B}", from=2-1, to=2-2]
	\arrow[""{name=1, anchor=center, inner sep=0}, "{X'' \ltie b}"{description}, from=2-2, to=3-2]
	\arrow[""{name=2, anchor=center, inner sep=0}, "{\J(f'f) \rtie B}", curve={height=-6pt}, from=1-1, to=2-2]
	\arrow[""{name=3, anchor=center, inner sep=0}, "{X'' \ltie (b'b)}", curve={height=-6pt}, from=2-2, to=3-3]
	\arrow["{\lc{\J f'}_b}"{description}, draw=none, from=0, to=1]
	\arrow["\iso"{description}, draw=none, from=2-1, to=2]
	\arrow["\iso"{description}, draw=none, from=3-2, to=3]
\end{tikzcd}
\]
\end{wrapfigure}
We get a left action 
	$\lact : \V \times \B \to \V$
 as follows. 
 On objects, we set 
 	$X \lact B := X \tens B$.
The action on 1-cells is
	$f \lact b := 
		\big(
			XB 
				\xra{\J(f) \rtie B}
			X'B 
				\xra{X' \ltie b}
			X'B'
		\big)$ 
with the evident action on 2-cells. The unitor is constructed from the unitors for the premonoidal structure, 
but the compositor relies on centrality. We define $\phi_{x, b}$ as on the right, where we write just $\iso$ for the compositors.

Next note that 
	$X \lact b = (X \ltie b) \circ (\J\Id_X \rtie B)$
so the unitor also gives a canonical structural isomorphism 
	$(X \lact b) \iso (X \ltie b)$
yielding an icon 
	$(X \lact -) \To (X \ltie -)$. 
So we may define the unitor to be the composite
	$\actlambda := (\tensu \lact -) \XRA\iso (\tensu \ltie -) \XRA\lambda \id$.
For the associator, we take the 1-cell components to be as for the premonoidal structure in $\B$, so that
	$\actalpha_{X, Y, C} := \alpha_{X,Y,C}$,
and define the 2-cell components using  $\theta, \zeta$, and the associator for the premonoidal structure: 
\begin{equation*}
		\cellOf{\actalpha}_{f,g,c} :=
\begin{tikzcd}[column sep = 5em, row sep =1.3em]
	{(XY)C} && {(X'Y')C} & {(X'Y')C'} \\
	& {(X'Y)C} \\
	{X(YC)} & {X'(YC)} & {X'(Y'C)} & {X'(Y'C')} 
	\arrow[""{name=0, anchor=center, inner sep=0}, from=1-1, to=2-2, curve={height=-12pt}]
	\arrow["{\J(f \tens g) \rtie C }", from=1-1, to=1-3]
	\arrow[""{name=1, anchor=center, inner sep=0}, ""'{yshift=-1mm, xshift=1mm, swap}, from=2-2, to=1-3,curve={height=-12pt}] %
	\arrow[""{name=2, anchor=center, inner sep=0}, "\alpha"', from=1-1, to=3-1]
	\arrow["{\J(f) \rtie (YC)}"', from=3-1, to=3-2]
	\arrow["{(X'Y')\ltie a}", from=1-3, to=1-4]
	\arrow["{X' \ltie (\J(g) \rtie C)}"', from=3-2, to=3-3]
	\arrow["{X' \rtie (Y' \rtie a)}"', from=3-3, to=3-4]
	\arrow[""{name=3, anchor=center, inner sep=0}, "\alpha", from=1-4, to=3-4]
	\arrow[""{name=4, anchor=center, inner sep=0}, from=1-3, to=3-3]
	\arrow[from=2-2, to=3-2]
	\arrow[""{name=5, anchor=center, inner sep=0}, curve={height=12pt}, from=2-2, to=1-3, "(X\J(g))C"{swap, yshift=1mm}] %
	\arrow[""{name=6, anchor=center, inner sep=0}, curve={height=12pt}, from=1-1, to=2-2, "{(\J (f)Y)C}"'{}]
	\arrow["{\cellOf{\alpha}_{X', Y', c}}"{description, xshift=-2mm}, draw=none, from=4, to=3]
	\arrow["{\cellOf{\alpha}_{\J f, Y, C}}"{description, yshift=-2mm}, shift right=5, draw=none, from=2, to=2-2]
	\arrow["{\cellOf{\alpha}_{X', \J g, C}}"{description, pos=0.6,yshift=-2mm}, shift right=5, draw=none, from=2-2, to=4]
	\arrow["\iso"{description, yshift =-2mm}, draw=none, from=0, to=1]
	\arrow["{\theta C }"{description}, draw=none, from=1, to=5]
	\arrow["{\zeta C}"{description}, draw=none, from=0, to=6]
\end{tikzcd}
\end{equation*}

The compatibility laws on 
	$\cellOf{\actlambda}$
and 
	$\cellOf{\actalpha}$
hold by the corresponding compatibility laws of a Freyd bicategory.
Turning now to the structural modifications, because the structural transformations agree with those of $\B$ on 1-cells, we take the corresponding modifications for the premonoidal structure. 
Showing these are indeed modifications relies on  the condition that 
	$\lc{\J f}_{\J g} = (\rc{\J g}_{\J f})^{-1}$. 
Consider the case of $\acttrianglem$. As 2-cells, 
	$\acttrianglem_{A,B} = \mathfrak{m}_{A,B}$
but $\acttrianglem$ is required to be a modification in two arguments, while the axioms of a premonoidal bicategory make
	$\mathfrak{m}_{A,B}$
a modification in each argument separately: in one argument, using $\rc{\lambda}$, and in the other argument using $\lc{\rho}$. 
Unpacking the equations for showing $\acttrianglem$ is a modification at maps
	$a : A \to A'$
and 
	$x : X \to X'$, 
we get an instance of $\lc{\J a}_{\lambda}$ arising from the compositor for $\lact$. To apply the modification law for $\mathfrak{m}$, therefore, we first need to pass through the equality
	$\lc{\J a}_{\lambda} = \lc{\J a}_{\J\lambda} = (\rc{\J\lambda}_{\J a})^{-1}$. 

The axioms of an action hold immediately from the axioms of a premonoidal bicategory.
The proof for the right action case is analogous, except one sets 
	$
	a \ract g := 
		\big( 
				AY \xra{a \rtie Y} A'Y \xra{A' \ltie \J g} A'Y'
		\big)
	$
and defines the compositor using right centrality. 
In summary, therefore, we have the following. 

\begin{proposition}
\label{res:freyd-bicategory-determines-two-actions}
Every Freyd bicategory
	$(\V \xra{\J} \B)$
determines a left action 
	$\lact : \V \times \B \to \B$
and a right action 
	$\ract : \B \times \V \to \V$.
\end{proposition}

\paragraph*{From Freyd bicategory to Freyd action.}

It remains to show the actions just constructed extend the canonical action of $\V$ on itself, and show they are compatible. 
\begin{wrapfigure}[5]{r}{6cm}
\vspace{-4mm}
\centering
\begin{tikzcd}
	& {X'Y} \\
	XY && {X'Y'}
	\arrow[""{name=0, anchor=center, inner sep=0}, "{\J(f) \rtie Y}", curve={height=-12pt}, from=2-1, to=1-2]
	\arrow[""{name=1, anchor=center, inner sep=0}, "{X' \ltie \J(g)}", curve={height=-12pt}, from=1-2, to=2-3]
	\arrow[""{name=2, anchor=center, inner sep=0}, "{\J(f \tens g)}"', from=2-1, to=2-3]
	\arrow[""{name=3, anchor=center, inner sep=0}, curve={height=12pt}, from=2-1, to=1-2]
	\arrow[""{name=4, anchor=center, inner sep=0}, curve={height=12pt}, from=1-2, to=2-3]
	\arrow["\zeta"{description}, draw=none, from=0, to=3]
	\arrow["\theta"{description}, draw=none, from=1, to=4]
	\arrow["\iso"{description}, draw=none, from=1-2, to=2]
\end{tikzcd}
\end{wrapfigure}
First we define icons $\theta'$ and $\zeta'$ by noting that 
$
	f \lact \J(g) 
	=
	(X' \ltie \J(g)) \circ (\J(f) \rtie Y)
	=
	\J(f) \ract g 
$
so that we can set $\theta'_{f,g}$ and $\zeta'_{f,g}$ both to be the composite diagram on the right.
In particular, $\theta'_{f, Y}$ and $\zeta'_{X, g}$ are just $\theta_f$ and $\zeta_g$, respectively, composed with structural isomorphisms. 
	
Now we define $\kappa$. On 1-cells we take just $\alpha$, but on 2-cells we take a definition similar to the proof of naturality in the 1-dimensional case:
for 
	$f : X \to X'$
and 
	$h : Z \to Z'$
in $\V$ and 
	$b : B \to B'$
in $\B$ we take: 
\begin{equation*}
\cellOf{\kappa}_{f,b,h} :=
\begin{tikzcd}[column sep = 4em, row sep = 1em]
	{(XB)Z} && {(X'B')Z} & {(X'B')Z'} \\
	& {(X'B)A} \\
	{X(BZ)} & {X'(BZ)} & {X'(B'Z)} & {X'(B'Z')} 	
	\arrow[
		swap,
		rounded corners,
		to path=
		{ -- ([yshift=.6cm]\tikztostart.north)
		-| ([yshift=.6cm]\tikztotarget.south)
		-- ([yshift=.7cm]\tikztotarget.south)
		}, 
		from=1-1, to=1-4
	]
	\arrow[
			swap,
			rounded corners,
			to path=
			{ -- ([yshift=-.4cm]\tikztostart.south)
			-| ([yshift=-.4cm]\tikztotarget.south)
			-- ([yshift=0cm]\tikztotarget.south)
			}, 
			from=3-2, to=3-4
	]
	\arrow[
			swap,
			rounded corners,
			to path=
			{ -- ([yshift=-.9cm]\tikztostart.south)
			-| ([yshift=-.9cm]\tikztotarget.south)
			-- ([yshift=0cm]\tikztotarget.south)
			}, 
			from=3-1, to=3-4
	]
	\arrow[from=3-2, to=3-4, "{X' \ltie (b \ract h)}", yshift=-1.2cm, draw=none]
	\arrow[from=3-2, to=3-4, "{\iso}", yshift=-.6cm, draw=none]
	\arrow[from=3-1, to=3-4, "{f \lact (b \ract h)}", yshift=-1.7cm, draw=none]
	\arrow[from=1-1, to=1-4, "{(f \lact b) \ract h}", yshift=1cm, draw=none]
	\arrow[""{name=0, anchor=center, inner sep=0}, "{(\J (f)B)Z}"{description, xshift=-2mm}, from=1-1, to=2-2]
	\arrow["{(f \lact b) \rtie Z }", from=1-1, to=1-3]
	\arrow[""{name=1, anchor=center, inner sep=0}, "{(X \ltie b)Z}"{description, xshift=2mm}, from=2-2, to=1-3]
	\arrow[""{name=2, anchor=center, inner sep=0}, "\alpha"', from=1-1, to=3-1]
	\arrow["{\J(f) \rtie (BZ)}"', from=3-1, to=3-2]
	\arrow["{(X'B')\ltie \J h}", from=1-3, to=1-4]
	\arrow["{X' \ltie (b \rtie Z)}"', from=3-2, to=3-3]
	\arrow["{X' \rtie (B' \rtie \J h)}"', from=3-3, to=3-4]
	\arrow[""{name=3, anchor=center, inner sep=0}, "\alpha", from=1-4, to=3-4]
	\arrow[""{name=4, anchor=center, inner sep=0}, from=1-3, to=3-3]
	\arrow[from=2-2, to=3-2]
	\arrow["{\cellOf{\alpha}_{X', B', \J h}}"{description}, draw=none, from=4, to=3]
	\arrow["\iso"{description}, draw=none, from=0, to=1]
	\arrow["{\cellOf{\alpha}_{\J f, B, Z}}"{description}, yshift=-3mm, draw=none, from=2, to=2-2]
	\arrow["{\cellOf{\alpha}_{X', b, Z}}"{description}, yshift=-3mm, draw=none, from=2-2, to=4]
\end{tikzcd}
\end{equation*}
	
The rest of the equations to check for the Freyd action are proven by applying the various compatibility laws to massage the statement into the corresponding axiom given by the definition of a Freyd bicategory. 
This completes the proof of the following.

\begin{proposition}
\label{res:Freyd-bicategory-to-Freyd-action}
Every Freyd bicategory
	$(\V \xra{\J} \B)$
determines a Freyd action with the same underlying pseudofunctor.
\end{proposition}

\subsection{The correspondence theorem}

\CorrespondenceTheorem*
\begin{proof}
We define functors 
	$F :\FreydAct{\V, \J, \B}
		\leftrightarrows
		\FreydBicat{\V, \J, \B} : G$
given on objects by the constructions in 
	\Cref{res:Freyd-action-to-Freyd-bicategory}
and 
	\Cref{res:Freyd-bicategory-to-Freyd-action}
respectively. 
So suppose 
	$(\lefttrans, \righttrans)$
is a map in $\FreydAct{\V, \J, \B}$. Then 
	$F(\lefttrans, \righttrans) := (F\lefttrans, F\righttrans)$
is defined by taking
\begin{align*}
	(F\lefttrans)^A_{f} 
		:= 
		\big( 
			(A \ltie f)  
			= (\Id_A \lact f) 
			\XRA{\lefttrans_{\Id_A, f}} (\Id_A \lact' f) 
			= (A \ltie' f) 
		\big) \\
	(F\righttrans)^A_{f} 
		:= 
		\big(
			(f \rtie A)  
			= (f \ract \Id_A) 
			\XRA{\righttrans_{f, \Id_A}} (f \ract' \Id_A) 
			= (f \rtie A)
		\big) 
\end{align*}
Conversely, given a map 
	$(\lefttrans, \righttrans)$
in $\FreydBicat{\V, \J, \B}$ we define 
	$G(\lefttrans, \righttrans) := (G\lefttrans, G\righttrans)$
to be 
\begin{equation*}
(G\lefttrans)_{f, b} :=
\begin{tikzcd}[column sep = 4em]
	XB & {X'B} & {X'B'}
	\arrow[""{name=0, anchor=center, inner sep=0}, curve={height=-12pt}, from=1-1, to=1-2, 
					"{\J f \rtie B}"]
	\arrow[""{name=1, anchor=center, inner sep=0}, curve={height=12pt}, from=1-1, to=1-2,
					"{\J f \rtie' B}"{swap}]
	\arrow[""{name=2, anchor=center, inner sep=0}, curve={height=-12pt}, from=1-2, to=1-3,
					"X' \ltie b"]
	\arrow[""{name=3, anchor=center, inner sep=0}, curve={height=12pt}, from=1-2, to=1-3,
					"X' \ltie' b"{swap}]
	\arrow[shorten <=2pt, shorten >=2pt, Rightarrow, from=2, to=3,
					"\:\:\lefttrans^{X'}_{b}"]
	\arrow[shorten <=2pt, shorten >=2pt, Rightarrow, from=0, to=1,
					"\:\:\righttrans^{B}_{\J f}"]
	\arrow[
			swap,
			rounded corners,
			to path=
			{ -- ([yshift=.8cm]\tikztostart.north)
			-| ([yshift=.8cm]\tikztotarget.south)
			-- ([yshift=.6cm]\tikztotarget.south)
			}, 
			from=1-1, to=1-3
		]
	\arrow[
			swap,
			rounded corners,
			to path=
			{ -- ([yshift=-.8cm]\tikztostart.south)
			-| ([yshift=-.8cm]\tikztotarget.north)
			-- ([yshift=-.6cm]\tikztotarget.north)
			}, 
			from=1-1, to=1-3
	]
	\arrow["", draw = none, from=1-1, to=1-3, "f \lact b", yshift = 1.1cm]
	\arrow["", draw = none, from=1-1, to=1-3, "\iso", yshift = .5cm]	
	\arrow["", draw = none, from=1-1, to=1-3, "\iso", yshift = -.9cm]		
	\arrow["", draw = none, from=1-1, to=1-3, "f \lact' b", yshift = -1.6cm]
\end{tikzcd}
\quad , \quad
(G\righttrans)_{a, g} :=
\begin{tikzcd}[column sep = 4em]
	AX & {A'X} & {A'X'}
	\arrow[""{name=0, anchor=center, inner sep=0}, curve={height=-12pt}, from=1-1, to=1-2, 
					"{a \rtie X}"]
	\arrow[""{name=1, anchor=center, inner sep=0}, curve={height=12pt}, from=1-1, to=1-2,
					"{a \rtie' X}"{swap}]
	\arrow[""{name=2, anchor=center, inner sep=0}, curve={height=-12pt}, from=1-2, to=1-3,
					"A' \ltie \J g"]
	\arrow[""{name=3, anchor=center, inner sep=0}, curve={height=12pt}, from=1-2, to=1-3,
					"A' \ltie' \J g"{swap}]
	\arrow[shorten <=2pt, shorten >=2pt, Rightarrow, from=2, to=3,
					"\:\:\lefttrans^{A'}_{\J g}"]
	\arrow[shorten <=2pt, shorten >=2pt, Rightarrow, from=0, to=1,
					"\:\:\righttrans^{X}_a"]
	\arrow[
			swap,
			rounded corners,
			to path=
			{ -- ([yshift=.8cm]\tikztostart.north)
			-| ([yshift=.8cm]\tikztotarget.south)
			-- ([yshift=.6cm]\tikztotarget.south)
			}, 
			from=1-1, to=1-3
		]
	\arrow[
			swap,
			rounded corners,
			to path=
			{ -- ([yshift=-.8cm]\tikztostart.south)
			-| ([yshift=-.8cm]\tikztotarget.north)
			-- ([yshift=-.6cm]\tikztotarget.north)
			}, 
			from=1-1, to=1-3
	]
	\arrow["", draw = none, from=1-1, to=1-3, "a \ract g", yshift = 1.1cm]
	\arrow["", draw = none, from=1-1, to=1-3, "\iso", yshift = .5cm]	
	\arrow["", draw = none, from=1-1, to=1-3, "\iso", yshift = -.9cm]		
	\arrow["", draw = none, from=1-1, to=1-3, "a \ract' g", yshift = -1.6cm]
\end{tikzcd}
\end{equation*}
One shows both $F$ and $G$ are well-defined by a long calculation using the compatibility properties on one side to show the required compatibility condition on the other side.

We now show that $GF \iso \id$ and $FG \iso \id$. Given an action 
	$\mathcal{A} := (\lact, \theta, \ract, \zeta, \kappa)$,
the composite $GF(\mathcal{A})$ has left action $\lact'$ given by
$
	f \lact' b 
	=
	(\Id_{X'} \lact b) \circ (\J f \ract \Id_B) 
$ 
and right action 
	$\ract'$
given by
$
	a \ract' g 
	=
	(\Id_{A'} \lact \J g) \circ (a \ract \Id_Y) 
$ 
so we get an obvious choice of icons 
	$\lact' \To \lact$
and 
	$\ract' \To \ract$
given by
\[
\begin{tikzcd}[column sep = 3em, row sep=2.5em]
	XB & {X'B} & {X'B'}
	\arrow[""{name=0, anchor=center, inner sep=0}, from=1-1, to=1-2]
	\arrow["{\Id_{X'} \lact b}", from=1-2, to=1-3]
	\arrow[""{name=1, anchor=center, inner sep=0}, "{\J f \ract \Id_B}", curve={height=-16pt}, from=1-1, to=1-2]
	\arrow[""{name=2, anchor=center, inner sep=0}, "{f \lact \Id_B}"', curve={height=16pt}, from=1-1, to=1-2]
	\arrow["\zeta"{description}, draw=none, from=1, to=0]
	\arrow["\theta"{description}, draw=none, from=0, to=2]
	\arrow[
			swap,
			rounded corners,
			to path=
			{ -- ([yshift=-.8cm]\tikztostart.south)
			-| ([yshift=-.8cm]\tikztotarget.north)
			-- ([yshift=-.6cm]\tikztotarget.north)
			}, 
			from=1-1, to=1-3
	]	
	\arrow["", draw = none, from=1-1, to=1-3, "\iso", yshift = -.9cm, xshift=4mm]		
	\arrow["", draw = none, from=1-1, to=1-3, "f \lact b", yshift = -1.6cm]
\end{tikzcd}
\qquad 
\qquad 
\begin{tikzcd}[column sep = 3em, row sep=2.5em]
	AY & {A'Y} & {A'Y'}
	\arrow["{a \ract \Id_{Y}}"', from=1-1, to=1-2]
	\arrow[""{name=0, anchor=center, inner sep=0}, "{\Id_{A'} \lact \J g}", curve={height=-16pt}, from=1-2, to=1-3]
	\arrow[""{name=1, anchor=center, inner sep=0}, "{\Id_{A'} \ract g}"', curve={height=16pt}, from=1-2, to=1-3]
	\arrow[""{name=2, anchor=center, inner sep=0}, from=1-2, to=1-3]
	\arrow["\theta"{description}, draw=none, from=0, to=2]
	\arrow["\zeta"{description}, draw=none, from=2, to=1]
	\arrow[
			swap,
			rounded corners,
			to path=
			{ -- ([yshift=-.8cm]\tikztostart.south)
			-| ([yshift=-.8cm]\tikztotarget.north)
			-- ([yshift=-.6cm]\tikztotarget.north)
			}, 
			from=1-1, to=1-3
	]	
	\arrow["", draw = none, from=1-1, to=1-3, "\iso", yshift = -.9cm, xshift=-4mm]		
	\arrow["", draw = none, from=1-1, to=1-3, "a \ract g", yshift = -1.6cm]
\end{tikzcd}
\]
These commute with all the data because $\theta$ and $\zeta$ do, and forms a natural isomorphism 
	$GF(\mathcal{A}) \iso \mathcal{A}$ 
because morphisms in $\FreydAct{\V, \J, \B}$ commute with the icons of the actions. 

Finally, to show that $FG \iso \id$ consider a Freyd bicategory 
	$\freydCat := (\ltie, \theta, \rtie, \zeta)$.
Then 
	$FG(\freydCat)$
has 
	${a \rtie' B} := (X' \ltie \J\Id_B) \circ (a \rtie B)$
and
	${A \ltie' b} := (A \ltie b) \circ (\J\Id_A \rtie B)$
so we have evident structural isomorphisms
	$(a \rtie' B) \iso (a \rtie B)$
and 
	$(A \ltie' b) \iso (A \ltie b)$. 
These commute with all the data and define a natural isomorphism 
	$FG(\freydCat) \iso \freydCat$
by straightforward applications of coherence.
\end{proof}